\documentclass[11pt]{article}

\usepackage{amsmath,amssymb,amsthm}
\usepackage{fullpage}
\usepackage{framed}
\usepackage{authblk}
\usepackage{hyperref}
\usepackage[position=top]{subfig}
\usepackage{graphicx,wrapfig}

\newtheorem{theorem}{Theorem}
\newtheorem{corollary}{Corollary}
\newtheorem{problem}{Problem}
\newtheorem{lemma}{Lemma}
\newtheorem{myclaim}{Claim}
\newtheorem{conjecture}{Conjecture}
\newtheorem{remark}{Remark}

\newcommand{\qedclaim}{\hfill $\diamond$ \medskip}
\newenvironment{proofclaim}{\noindent\ignorespaces{\em Proof.}}{\hfill\qedclaim\par\noindent} 

\usepackage{color}

\newsavebox{\mybox}

\newcounter{subfloat}
\renewcommand{\thesubfloat}{\alph{subfloat}}
\newcommand{\image}[2]{%
  \stepcounter{subfloat}%
  \begin{tabular}[t]{@{}c@{}}
  #2 \\
  (\thesubfloat) #1
  \end{tabular}%
}

\begin{document}

\title{A story of diameter, radius and Helly property}
\author[1,2]{Guillaume Ducoffe}
\author[3]{Feodor F. Dragan}
\affil[1]{\small National Institute for Research and Development in Informatics, Romania}
\affil[2]{\small University of Bucharest, Romania}
\affil[3]{\small Computer Science Department, Kent State University, Kent, Ohio, USA}
\date{}

\maketitle

\begin{abstract}
A graph is Helly if every family of pairwise intersecting balls has a nonempty common intersection.
Motivated by previous work on dually chordal graphs and graphs of bounded distance VC-dimension (with the former being a subclass of Helly graphs and the latter being a particular case of graphs of bounded fractional Helly number, respectively) we prove several new results on the complexity of computing the diameter and the radius on Helly graphs and related graph classes.
\begin{itemize}
\item First, we present  algorithms which given an $n$-vertex $m$-edge Helly graph $G$ as input, compute w.h.p. its radius and its diameter in time $\tilde{\cal O}(m\sqrt{n})$. Our algorithms are based on the Helly property and on several implications of the unimodality of the eccentricity function in Helly graphs: every vertex of locally minimum eccentricity is a central vertex.
\item Then, we focus on $C_4$-free Helly graphs, which include, amongst other subclasses, bridged Helly graphs and so, chordal Helly graphs and hereditary Helly graphs. For the $C_4$-free Helly graphs, we present linear-time algorithms for computing the eccentricity of all vertices. Doing so, we generalize previous results on strongly chordal graphs to a much larger subclass.
\item Finally, we derive from our findings on chordal Helly graphs a more general one-to-many reduction from diameter computation on chordal graphs to either diameter computation on split graphs or the {\sc Disjoint Set} problem. Therefore, split graphs are in some sense the {\em only} hard instances for diameter computation on chordal graphs. As a byproduct of our reduction  the eccentricity of all vertices in a chordal graph can be approximated in ${\cal O}(m\log{n})$ time with an additive one-sided error of at most one. This answers an open question of [Dragan, IPL 2019]. We also get from our reduction that on any subclass of chordal graphs with constant VC-dimension the diameter can be computed in truly subquadratic time. 
\end{itemize}
These above results are a new step toward better understanding the role of abstract geometric properties in the fast computation of metric graph invariants.
\end{abstract}

\section{Introduction}\label{sec:introduction}

For any undefined graph terminology, see~\cite{BoM08,Die12}.
We study the fundamental problems of computing, for a given undirected unweighted graph, its {\em diameter} and its {\em radius}.
There is a textbook algorithm for solving both problems in ${\cal O}(nm)$ time on $n$-vertex $m$-edge graphs.
However, it is a direct reduction to All-Pairs Shortest-Paths (APSP), that is a seemingly more complex problem with a much larger (quadratic-size) output than for the diameter and radius problems.
On one hand, there is a long line of work presenting more efficient -- often linear-time -- algorithms for computing the diameter and/or the radius on some special graph classes~\cite{AVW16,BCD98,BHM18,Cab18,ChD94,CDV02,CDHP01,CDP18,Dam16,DrN00,Dra89,Duc19,DHV19+b,DHV19+a,Epp00,FaP80,GKHM+18,Ola90}.
On the other hand, under the Strong Exponential-Time Hypothesis (SETH) and the Hitting Set Conjecture (HS), respectively, we cannot solve either of these two problems in truly subquadratic-time~\cite{AVW16,RoV13}\footnote{
By truly subquadratic we mean a running-time in ${\cal O}(n^am^b)$, for some positive $a,b$ such that $a + b< 2$.}.

We aim at characterizing the graph classes for which these above (SETH- or HS-) ``hardness'' results do not hold.
Ideally we would like to derive a dichotomy theorem, not unlike those proved in~\cite{AVW16} but covering many more subquadratic-time solvable special cases from the literature.
Our work is part of a recent series of papers, with co-authors, where we try to reach this objective based on tools and concepts from Computational Geometry~\cite{Duc19,DHV19+b,DHV19+a}.
-- See also~\cite{AVW16,Cab18} for some pioneering works in this line of research. --
Specifically, a class ${\cal H}$ of hypergraphs has {\em fractional Helly number} at most $k$ if for any positive $\alpha$ there is some positive $\beta$ such that, in any subfamily of hyperedges in a hypergraph of ${\cal H}$, if there is at least a fraction $\alpha$ of all the $k$-tuples of hyperedges with a non-empty common intersection,  then there exists an element that is contained in a fraction at least $\beta$ of all hyperedges in this subfamily. 
Then, the fractional Helly number of a graph class ${\cal G}$ is the fractional Helly number of the family of the ball hypergraphs of all graphs in ${\cal G}$.
For instance Matousek proved that every hypergraph of {\em VC-dimension} $d$ has its fractional Helly number that is upper bounded by a function of $d$~\cite{Mat04}.
It implies that the graphs of bounded distance VC-dimension, studied in~\cite{BoC14,BoT15,CEV07,DHV19+a}, have a bounded fractional Helly number.
Note that the latter graphs comprise planar graphs~\cite{CEV07}, bounded clique-width graphs~\cite{BoT15} and interval graphs~\cite{DHV19+a} amongst other subclasses of interest.
Motivated by the results from~\cite{DHV19+a} on diameter and radius computation in these graphs, we ask whether we can compute the diameter and the radius in truly subquadratic time within graph classes of constant fractional Helly number.
As a first step toward resolving this question, our current research focuses on the simpler class of {\em Helly graphs}.
Recall that a graph is Helly if any family of pairwise intersecting balls has a non-empty common intersection.
So, in particular, Helly graphs have fractional Helly number two.
Furthermore, we stress that the Helly graphs are one of the most studied classes in Metric Graph Theory ({\it e.g.}, see the survey~\cite{BaC08} and the papers cited therein). 
 Indeed, this is partly because of the very nice property that every graph is an isometric subgraph of some Helly graph, thereby making of the latter the discrete equivalent of hyperconvex metric spaces~\cite{Dre84,Isb64}.

\begin{conjecture}\label{conj:helly}
There is a positive $\varepsilon$ such that, for every $n$-vertex $m$-edge Helly graph, we can compute its diameter and its radius in time ${\cal O}(mn^{1-\varepsilon})$.
\end{conjecture}

To the best of our knowledge, until this work the only known result  toward proving Conjecture~\ref{conj:helly} was a positive answer to this conjecture for the class of dually chordal graphs~\cite{BCD98}.
In fact, the diameter and the radius of a dually chordal graph can be computed in {\em linear time}, that is optimal. We note that obtaining such a strong complexity result for the general Helly graphs looks more challenging. Hence, we also studied whether stronger versions of Conjecture~\ref{conj:helly} could hold true
on more restricted subclasses, such as {\em chordal} Helly  graphs or more generally $C_4$-free Helly graphs.
-- This latter choice was also partly motivated by a nice characterization of {\em hereditary} Helly graphs: indeed, they are exactly the $3$-sun-free chordal graphs~\cite{Dra89}. --
We stress that it is already SETH-hard to compute the diameter on chordal graphs in truly subquadratic time~\cite{BCH16}.
On the positive side, there exist linear-time algorithms for computing the radius on general chordal graphs~\cite{ChD94}, and the diameter on various subclasses of chordal graphs, {\it e.g.} interval graphs, directed path graphs and strongly chordal graphs~\cite{BCD98,CDHP01,DNB97}.
Most of these special cases, including the three aforementioned examples, are strict subclasses of chordal Helly graphs.
As a result, our work pushes forward the tractability border for diameter computation on chordal graphs and beyond.

\subsection*{Our Contributions}
Our first main result in the paper is a  positive resolution of Conjecture~\ref{conj:helly} in the general case. Specifically, we present truly subquadratic-time algorithms for computing both the radius and the diameter of Helly graphs (Theorem~\ref{thm:helly-rad} and Theorem~\ref{thm:helly-diam}). Although for a Helly graph we can compute its diameter from its radius~\cite{Dra89} -- a property which is not known to hold for general graphs -- we present separate algorithms for diameter and radius computations. Indeed, our approach for computing the radius can be applied to a broader class than the Helly graphs, both as an exact and approximation algorithm. 
 Our algorithms run in time $\tilde{\cal O}(m\sqrt{n})$ w.h.p., and they use as their main ingredients several consequences of the unimodality of the eccentricity function in Helly graphs~\cite{Dra89}: every local minimum of the eccentricity function in a Helly graph is a global minimum.

\smallskip
Next we focus on the class of $C_4$-free Helly graphs, which have been studied on their own and have more interesting convexity properties than general Helly graphs~\cite{Dra89,Dra93}. In particular, the center $C(G)$ of a $C_4$-free Helly graph $G$ is convex and it has diameter at most 3 and radius at most 2~\cite{Dra89,Dra93}.  In contrast, the center $C(G)$ of a general Helly graph $G$ is isometric but it can have arbitrarily large diameter; in fact, any Helly graph $H$ is the center of some other Helly graph $G$~\cite{Dra89}. 
We stress that $C_4$-free Helly graphs encompass the bridged Helly  graphs and all Helly graphs of hyperbolicity $1/2$, amongst other examples.
By restricting ourselves to this subclass we can use the well-known multi-sweep heuristic of Corneil et al.~\cite{CDHP01}, in order to compute vertices of provably large eccentricity, as a brick-basis for {\em exact} linear-time algorithms for computing both a central vertex and the diameter.
Our general approach for these graphs is also partly inspired by the algorithms of Chepoi and Dragan~\cite{ChD94} and Dragan and Nicolai~\cite{DrN00}, in order to compute a central vertex in chordal graphs and a diametral pair in distance-hereditary graphs, respectively.
We stress that in contrast to this positive result on $C_4$-free Helly graphs, and as notified to us by Chepoi (private communication), a similar method cannot apply to general Helly graphs.
Indeed, the values obtained for Helly graphs with the multi-sweep heuristic can be arbitrarily far from the diameter, which comes from the property that any graph can be isometrically embedded into a Helly graph~\cite{Dre84,Isb64}.

Furthermore, our results for $C_4$-free Helly graphs go beyond the mere calculation of the diameter and the radius. Indeed, we are able to compute the eccentricity of all vertices, which for Helly graphs can be reduced to computing the graph center. For that, we first need to solve the related problem of computing a {\em diametral pair} ({\it i.e.}, a pair of vertices of which the distance in the graph equals the diameter), which surprisingly requires a more intricate approach than for just computing the diameter. -- We note that the algorithm of Theorem~\ref{thm:helly-diam} also computes a diametral pair, but in super-linear time. -- This intermediate result has interesting consequences on its own. For instance, if we apply our algorithms on an {\em arbitrary} chordal graph, then we can use a (supposedly) diametral pair in order to decide, in linear time, if either we computed the diameter correctly or the input graph is not Helly. See Remark~\ref{rk:check-result} for more details.
Note that in comparison, the best-known recognition algorithms for chordal Helly graphs run in ${\cal O}(m^2)$ time~\cite{LiS07}. 
Our two main ingredients in order to solve these above problems are \texttt{(i)} a ``pseudo-gatedness'' property of the subsets of weak diameter at most two in $C_4$-free Helly graphs, and \texttt{(ii)} a reduction from finding a diametral pair under some technical assumptions to the same problem on a related {\em split} Helly graph.
We find the latter result all the more interesting that split graphs are amongst the hardest instances for diameter computation~\cite{BCH16}.
%

\smallskip
Finally, our above investigations on $C_4$-free Helly graphs lead us to the following natural research question: what are the other graph classes where the diameter can be efficiently computed from a subfamily of split graphs? In particular, can we reduce diameter computation on {\em general} chordal graphs to the same problem on split graphs?
This would imply that the subclass of split graphs is, in some sense, the sole hard case for diameter computation on chordal graphs.
Furthermore, this could help in finding new subclasses of chordal graphs for which we can compute the diameter faster than in ${\cal O}(nm)$. 
We answer positively to this open question, but in a more restricted setting (Theorem~\ref{thm:chordal-reduction}).
Specifically, our reduction is indeed from diameter computation on chordal graphs to the same problem on split graphs, but the computational results which we obtain are better if then we reduce to the well-known {\sc Disjoint Set} problem -- {\it a.k.a.} the monochromatic {\sc Orthogonal Vector} problem.
We stress that there is a trivial linear-time reduction from diameter computation on split graphs to {\sc Disjoint Set}, but the converse reduction from this problem to computing the diameter of a related split graph runs in time {\em quadratic} in the number of elements in the ground set of the input family.
This is evidence that {\sc Disjoint Set} might be a harder problem than diameter computation on split graphs -- at least in some density regimes.

As a byproduct of our reduction, we prove that the diameter can be computed in truly subquadratic time on any subclass of chordal graphs with constant {\em VC-dimension} (Theorem~\ref{thm:vc-dim}).
This nicely complements the results from~\cite{DHV19+a}, which mostly apply to sparse graph classes of constant distance VC-dimension or assuming a bounded (sublinear) diameter.
 Another application of our reduction is the approximate computation in quasi linear time of the eccentricity of all vertices in a chordal graph with an additive one-sided error of at most $1$. The latter result answers an open question from~\cite{Dra19}.

\paragraph{Notations.} 
Throughout the remainder of the paper, we denote by $dist_G(u,v)$ the distance between vertices $u$ and $v$.
The metric interval $I_G(u,v)$ between $u$ and $v$ is defined as $\{ w \in V \mid dist_G(u,w) + dist_G(w,v) = dist_G(u,v)\}$.
For any $k \leq dist_G(u,v)$, we can also define the slice $L(u,k,v) := \{ w \in I_G(u,v) \mid dist_G(u,w) = k\}$.
The ball of radius $r$ and center $v$ is defined as $\{ u \in V \mid dist_G(u,v) \leq r \}$, and denoted $N_G^r[v]$.
In particular, $N_G[v] := N_G^1[v]$ and $N_G(v) := N_G[v] \setminus \{v\}$ denote the closed and open neighbourhoods of a vertex $v$, respectively.
More generally, for any vertex-subset $S$ we define $dist_G(u,S) := \min_{v \in S} dist_G(u,v), \ N_G^r[S] := \bigcup_{v \in S}N_G^r[v], \ N_G[S] := N_G^1[S] \ \text{and} \ N_G(S) := N_G[S] \setminus S$.
The metric projection of a vertex $u$ on $S$, denoted $Pr_G(u,S)$, is defined as $\{ v \in S \mid dist_G(u,v) = dist_G(u,S) \}$.
The eccentricity of a vertex $u$ is defined as $\max_{v \in V} dist_G(u,v)$ and denoted by $e_G(u)$.
We also define the set $F_G(u) := \{ v \in V \mid dist_G(u,v) = e_G(u) \}$ of all the farthest vertices from vertex $u$.
-- Note that we will omit the subscript if the graph $G$ is clear from the context. -- 
The radius and the diameter of a graph $G$ are denoted $rad(G)$ and $diam(G)$, respectively.
Finally, $C(G) := \{ v \in V \mid e(v) = rad(G) \}$ is the center of $G$, {\it a.k.a.} the set of all the central vertices of $G$.

\section{Fast Computations within Helly graphs}\label{sec:gal-helly}

 In this section, we answer positively to Conjecture~\ref{conj:helly}. Section~\ref{sec:rad-helly} is devoted to a subquadratic-time randomized algorithm for radius computation, that can be turned to an exact or approximation algorithm for larger classes than the Helly graphs (namely, in every class where the diameter equals twice the radius, up to some additive constant). Then, we combine this approach with several other technical arguments in order to compute the diameter of Helly graphs in truly subquadratic time (Section~\ref{sec:diam-helly}).

 \subsection{Radius computation}\label{sec:rad-helly}

We start this section with a simple randomized test, which is inspired from previous works on adaptive greedy set cover algorithms~\cite{SeM10}.

\begin{lemma}\label{lem:set-cover}
Let $G=(V,E)$ be a graph, let $r$ be a positive integer and let $\varepsilon \in (0;1)$.
There is an algorithm that w.h.p. computes a set $D\langle G; r; \varepsilon \rangle$ in $\tilde{\cal O}(m/\varepsilon)$ time with the following two properties:
\begin{itemize}
\item if $e(v) \leq r$ then $v \in D\langle G; r; \varepsilon \rangle$;
\item conversely, if $v \in D\langle G; r; \varepsilon \rangle$ then  $|N^r[v]| \geq (1-\varepsilon) \cdot n$.
\end{itemize}
\end{lemma}

\begin{proof}
Let $p = c \cdot \frac{\log{n}}{\varepsilon n} $ for some arbitrary large constant $c$. If $p \geq 1$ then $n \leq c  \frac{\log{n}}{\varepsilon}$, and so we can compute the set of all the vertices of eccentricity at most $r$ in time $\tilde{\cal O}(m/\varepsilon)$ by running a BFS from every vertex.
From now on we assume that $p < 1$.
By $U(p)$ we mean a subset in which every vertex was added independently at random with probability $p$. 
Observe that we have $\mathbb{E}[|U(p)|] = c \frac{\log{n}}{\varepsilon} > c \cdot \log{n}$.
By Chernoff bounds we get $|U(p)| = \tilde{\cal O}(\varepsilon^{-1})$ with probability $\geq 1 - n^{-c}$. 
Then, for every $v \in V$, we compute $N^r[v] \cap U(p)$, which takes total time $\tilde{\cal O}(m/\varepsilon)$. 
We divide our analysis in two cases.
First let us assume that $e(v) \leq r$.
Then, with probability $1$ we have $U(p) \subseteq N^r[v]$. 
Second, let us assume that $|N^r[v]| < (1-\varepsilon) \cdot n$.
We get $Prob[ U(p) \subseteq N^r[v] ] < (1-p)^{\varepsilon n} = (1-p)^{\frac 1 p \cdot c \log{n}} \leq n^{-c}$.
Overall, let $D\langle G; r; \varepsilon \rangle$ contain all the vertices $v$ such that $U(p) \subseteq N^r[v]$.
By a union bound over $\leq n$ vertices, the set $D\langle G; r; \varepsilon \rangle$ satisfies our two above-stated properties with probability $\geq 1 - n^{-(c-1)}$.
\end{proof}

We derive from this simple test above an approximation algorithm for computing the radius and the diameter, namely:

\begin{lemma}\label{lem:approx-diam-rad}
Let $G=(V,E)$ be a graph and $r$ be a positive integer.
There is an algorithm that w.h.p. runs in $\tilde{\cal O}(m\sqrt{n})$ time and such that:
\begin{itemize}
\item If the algorithm accepts then $diam(G) \leq 2r$;
\item If the algorithm rejects then $rad(G) > r$.
\end{itemize}
\end{lemma}

Note that since $diam(G) \leq 2rad(G)$, this algorithm rejects any graph $G$ with $diam(G) > 2r$. However, it might also reject some graphs $G$ such that $diam(G) \leq 2r$ but $rad(G) > r$.

\begin{proof}
For some $\varepsilon$ to be defined later, we construct a set $D\langle G; r; \varepsilon \rangle$ as in Lemma~\ref{lem:set-cover}.
W.h.p. it takes time $\tilde{\cal O}(m/\varepsilon)$.
There are two cases.
If $D\langle G; r; \varepsilon \rangle = \emptyset$ then we know that $rad(G) > r$ and we stop.
Otherwise, we pick any vertex $c \in D\langle G; r; \varepsilon \rangle$ and we compute $N^r[c]$. 
Here it is important to observe that all the vertices of $N^r[c]$ are pairwise at a distance $\leq 2r$. 
Furthermore, w.h.p. we have $|V \setminus N^r[c]| \leq \varepsilon \cdot n$.
We end up computing a BFS from every vertex of $V \setminus N^r[c]$, accepting in the end if and only if all these vertices have eccentricity $\leq 2r$.
By setting $\varepsilon = n^{-1/2}$, the total running time is w.h.p. in $\tilde{\cal O}(m\sqrt{n})$. 
\end{proof}

An important consequence of Lemma~\ref{lem:approx-diam-rad} is that the hard instances for diameter and radius approximations are those for which the difference $2rad(G) - diam(G)$ is large, namely:

\begin{corollary}\label{cor:approx-diam-rad}
If $2rad(G) - diam(G) \leq k$ then, w.h.p., we can compute an additive $+\left\lfloor k/2 \right\rfloor$-approximation of $rad(G)$ and an additive $+k$-approximation of $diam(G)$ in total $\tilde{\cal O}(m\sqrt{n})$ time.
\end{corollary}

\begin{proof}
We compute by dichotomic search the smallest $r$ such that the algorithm of Lemma~\ref{lem:approx-diam-rad} accepts. Note that w.h.p. $r \geq \left\lceil diam(G)/2 \right\rceil$, and so $r \geq \left\lceil \frac{2rad(G)-k}2\right\rceil = rad(G) - \left\lfloor k/2 \right\rfloor$. Furthermore, we have w.h.p. $r \leq rad(G)$, and so $2r \leq diam(G) + k$. We output $r$ and $2r$ as approximations of $rad(G)$ and $diam(G)$, respectively. 
\end{proof}

\paragraph{Application to Helly graphs.}
For Helly graphs, the diameter and the radius are closely related.
This is a consequence of the unimodality property of the eccentricity function of Helly graphs~\cite{Dra89}, a property that will be further discussed in the next section.
In particular, the following relations hold between the two:

\begin{lemma}[\cite{Dra89}]\label{lem:unimodal}
If $G$ is a Helly graph then $2 rad(G) \geq diam(G) \geq 2 rad(G) - 1$.
In particular, $rad(G) = \left\lceil diam(G)/2 \right\rceil$.
\end{lemma}

By combining Lemma~\ref{lem:unimodal} with Corollary~\ref{cor:approx-diam-rad},  we obtain the main result of this section, namely:

\begin{theorem}\label{thm:helly-rad}
If $G$ is a Helly graph then, w.h.p., we can compute $rad(G)$ and an additive $+1$-approximation of $diam(G)$ in time $\tilde{\cal O}(m\sqrt{n})$.
\end{theorem}

\subsection{Diameter computation}\label{sec:diam-helly}

Equipped with Theorem~\ref{thm:helly-rad}, we already know how to compute, for a Helly graph, an additive $+1$-approximation of its diameter.
However, this is not enough yet in order to prove Conjecture~\ref{conj:helly}.
For instance, consider the case of chordal graphs: we can compute their radius~\cite{ChD94} and an additive $+1$-approximation of their diameter~\cite{CDHP01,DNB97} in linear time, however it is SETH-hard to compute their diameter exactly in truly subquadratic time~\cite{BCH16}.
Our main result in this section is that the exact diameter of Helly graphs can be computed in truly subquadratic time (Theorem~\ref{thm:helly-diam}), that is in sharp contrast with the chordal graphs.

\subsubsection*{An intermediate problem}

We start with a parameterized algorithm for computing all eccentricities up to some threshold value in a Helly graph $G$. Our results for the following more general problem are also used in Section~\ref{sec:all-ecc}.

\begin{center}
	\fbox{
		\begin{minipage}{.95\linewidth}
			\begin{problem}[\textsc{Small Eccentricities}]\
				\label{prob:small-ecc} 
					\begin{description}
					\item[Input:] a graph $G=(V,E)$; a vertex-subset $A$; a positive integer $k$.
					\item[Output:] the set $B_k := \{ b \in V \mid A \subseteq N^k[b] \}$.
				\end{description}
			\end{problem}     
		\end{minipage}
	}
\end{center}

We note that already for $k=2$, Problem~\ref{prob:small-ecc} is unlikely to be solvable in truly subquadratic time.
Indeed, this special case is somewhat related to the {\sc Hitting Set} problem~\cite{AVW16}. 
We explain next how to solve this problem in parameterized linear time when $G$ is a Helly graph. 

\begin{theorem}\label{thm:all-ecc-param}
If $G$ is Helly then, for every subset $A$ and every positive integer $k$, we can solve {\sc Small Eccentricities} in ${\cal O}(km)$ time.
\end{theorem}

\begin{proof}
We reduce the problem to the construction of a partition ${\cal P}_A^k = (A_1^{k},A_2^{k},\ldots,A_{p_k}^{k})$ of $A$ such that, for every $1 \leq j \leq p_k$, we have $\cap \{ N^{k}[a] \mid a \in A_j^{k} \} = V_j^{k} \neq \emptyset$, and furthermore the sets $V_1^{k},V_2^{k},\ldots,V_{p_k}^{k}$ are pairwise disjoint. Indeed, observe that we have $B_k \neq \emptyset$ if and only if $p_k = 1$, and in such a case $B_k = V_1^k$.
It now remains to prove that we can construct the partition ${\cal P}_A^{k}$, and the associated sets $V_j^{k}$, in ${\cal O}(km)$ time. If $k=0$, then we set ${\cal P}_A^0 = A$ and we are done (notice that for every $a_j \in A$, the corresponding set $V_j^0$ is exactly the singleton $\{a_j\}$). Otherwise, we show how to construct ${\cal P}_A^{k}$ and its associated sets from ${\cal P}_A^{k-1}$ and the $V_j^{k-1}$'s, in linear time. 

For that, let us define for every $j$ the new subset $W_j^{k} := N[V_j^{k-1}]$. Notice that constructing the sets $W_j^{k}$ takes total linear time as by the hypothesis, the sets $V_j^{k-1}$ are pairwise disjoint. However, we may have $W_j^k \cap W_{j'}^k \neq \emptyset$ for some $j \neq j'$.
\begin{myclaim}
$W_j^{k} = \bigcap \{ N^{k}[a] \mid a \in A_j^{k-1} \}$.
\end{myclaim}
\begin{proofclaim}
Since we have $V_j^{k-1} = \bigcap \{ N^{k-1}[a] \mid a \in A_j^{k-1} \}$, we get $W_j^{k} \subseteq \bigcap \{ N^{k}[a] \mid a \in A_j^{k-1} \}$ by construction. Conversely, let $v \in \bigcap \{ N^{k}[a] \mid a \in A_j^{k-1} \}$ be arbitrary. Since $N[v]$ and $N^{k-1}[a], \forall a \in A_j^{k-1}$ pairwise intersect, by the Helly property, $N[v] \cap V_j^{k-1} \neq \emptyset$, proving that $v  \in W_j^{k}$.
\end{proofclaim}
We are left with computing the ${V_j^k}'s$ as the union of pairwise disjoint subfamilies of the $W_{j'}^k$'s.
For that, we need the following additional result:
\begin{myclaim}
Let $v \in V$ be such that $\# \{ j \mid v \in W_j^{k} \}$ is maximized.
For any $j'$ such that $v \notin W_{j'}^{k}$, we have $W_{j'}^{k} \cap \bigcap \{ W_j^k \mid v \in W_j^k \} = \emptyset$.
\end{myclaim}
\begin{proofclaim}
Suppose for the sake of contradiction that there exists a vertex $u \in W_{j'}^{k} \cap \bigcap \{ W_j^k \mid v \in W_j^k \}$.
Then, we get $\{ j \mid v \in W_j^{k} \} \subset \{ j' \mid u \in W_{j'}^{k} \}$.
However, the latter contradicts the maximality of vertex $v$.
\end{proofclaim}
We now proceed as follows in order to compute ${\cal P}_A^{k}$ and the $V_j^k$'s.
Let ${\cal F}_k$ be an empty family of sets (at the end of this sub-procedure below, we shall get ${\cal F}_k = (V_1^k,V_2^k,\ldots,V_{p_k}^k)$).
While ${\cal P}_A^{k-1} \neq \emptyset$ we pick a vertex $v \in V$ such that $\# \{ j \mid v \in W_j^{k} \}$ is maximized.
We add the new sets $A_v := \bigcup \{ A_j^{k-1} \mid v \in W_j^{k} \}$ and $\bigcap \{ W_j^{k} \mid v \in  W_j^{k} \}$, in the families ${\cal P}_A^k$ and ${\cal F}_k$, respectively.
Indeed, $\bigcap \{ N^{k}[a] \mid a \in A_v \} = \bigcap \{ W_j^{k} \mid v \in  W_j^{k} \}$.
Then, we remove from ${\cal P}_A^{k-1}$ every $A_j^{k-1}$ such that $v \in W_j^{k}$. 
By the above claim, the sets in ${\cal F}_k$ are pairwise disjoint.

Finally, in order to construct ${\cal F}_k$, during a pre-processing step we compute $\# \{ j \mid v \in W_j^{k} \}$ for every vertex $v$. 
Since the $W_j^k$'s can be constructed in total linear time, this pre-processing also takes linear time.
We create an array of $|{\cal P}_A^{k-1}|$ lists, where $\forall i$ the $i^{\text{th}}$ list contains all the vertices that are in exactly $i$ subsets $W_j^k$.
Then, starting from $i = |{\cal P}_A^{k-1}|$, if the $i^{\text{th}}$ list is empty then $i := i-1$. Otherwise, we can pick any vertex $v$ of this list as it maximizes $\# \{ j \mid v \in W_j^{k} \}$. In this latter case the total running time of the step is in ${\cal O}(\sum_{j \mid v \in W_j^k} |A_j^{k-1}| + |W_j^k|)$. Since all the subsets $A_j^{k-1}$, such that $v \in W_j^k$, are subsequently removed from ${\cal P}_A^{k-1}$, after this step $v$ is no more contained in a group $W_j^k$, for any $A_j^{k-1} \in {\cal P}_A^{k-1}$ and so it will never be used again during the sub-procedure. Overall, the running time is in ${\cal O}(\sum_j |A_j^{k-1}| + |W_j^k|) = {\cal O}(n+m)$.
\end{proof}

\begin{corollary}\label{cor:all-ecc}
For any Helly graph $G$ and positive integer $k$, we can compute the set of all the vertices of eccentricity at most $k$ in ${\cal O}(km)$ time.
\end{corollary}

\begin{proof}
It suffices to apply Theorem~\ref{thm:all-ecc-param} with $A = V$.
\end{proof}

\subsubsection*{Using the unimodality of the eccentricity function}

We stress that using our previous Corollary~\ref{cor:all-ecc}, if the diameter of a Helly graph is sublinear in the number of nodes, then we can compute the eccentricity of all vertices in truly subquadratic time. However, there exist very simple Helly graphs, such as paths, for which the diameter is linear in the number of nodes. For such ``giant-diameter'' Helly graphs, we next adapt a well-known sampling technique for distance oracles~\cite{BCE05}.

Recall that a function is called {\em unimodal} if every its local minimum is global. 
It was proved in~\cite{Dra89} that the eccentricity function in Helly graphs is unimodal, and that the latter implies the following interesting property:

\begin{lemma}[\cite{Dra89}]\label{lem:ecc-formula}
If $G$ is Helly then, for every vertex $v$, $e(v) = dist(v,C(G)) + rad(G)$.
\end{lemma}

\begin{theorem}\label{thm:giant-diam}
Let $G$ be a Helly graph such that $rad(G) > 3k = \omega(\log{n})$. Then w.h.p. in $\tilde{\cal O}(mn/k)$ time, we can compute a diametral pair for $G$.
\end{theorem}

\begin{proof}
Let $p = c \log{n}/k$ for some sufficiently large constant $c$.
We construct a subset $U(p)$ where every vertex is included independently with probability $p$.
By Chernoff bounds we have $|U(p)| = \tilde{\cal O}(n/k)$ w.h.p., and we assume from now on that it is indeed the case.
In particular, we can perform a BFS from every vertex in $U(p)$ in $\tilde{\cal O}(mn/k)$ time.
Then, we define for every vertex $v \in V$: $$\bar{e}(v) := \min_{u \in U(p) \mid dist(u,v) \leq k} dist(u,v) + e(u)$$
(by convention, we set $\bar{e}(v) = 0$ if every vertex of $U(p)$ is at distance $> k$ from $v$).
Note that all the values $\bar{e}(v)$ can be computed in ${\cal O}(n|U(p)|) = \tilde{\cal O}(n^2/k)$ time.
We now divide our analysis in two cases:
\begin{itemize}
\item Case $e(v) < rad(G) + k$. Then, for any $u \in U(p)$ such that  $dist(u,v) \leq k$, we get $e(u) \leq e(v) + k < rad(G) + 2k$. Hence, $\bar{e}(v) < rad(G) + 3k \leq 2 \cdot rad(G) - 1$. By Lemma~\ref{lem:unimodal}, we get that $\bar{e}(v) < diam(G)$ with probability $1$.
\item Case $e(v) \geq rad(G) + k$. Note that in particular, we always fall in this case if $v$ is an end of a diametral path. Let us consider the set $S_v$ of the $k$ first vertices on a fixed shortest path between $v$ and a closest vertex of $C(G)$. By Lemma~\ref{lem:ecc-formula}, $e(v) = dist(v,C(G)) + rad(G)$. In particular, for every $u \in S_v$ we have $e(v) = dist(u,v) + e(u)$. If $v$ is an end of a diametral path, and furthermore $U(p) \cap S_v \neq \emptyset$, then this implies $\bar{e}(v) = e(v) = diam(G)$. Therefore, we are left proving that w.h.p., $U(p) \cap S_v \neq \emptyset$.
That is indeed the case since $Prob[ U(p) \cap S_v = \emptyset] = (1-p)^{|S_v|} = (1-p)^k = (1-p)^{\frac{c \log{n}}{p}} \leq n^{-c}$.
\end{itemize}
Finally, it follows from our above analysis that, w.h.p., a vertex $v$ which maximizes $\bar{e}(v)$ is an end of a diametral path. Once such an end is computed, we can compute a corresponding diametral pair in linear time, {\it e.g.,} using BFS.
\end{proof}

Combining Theorem~\ref{thm:all-ecc-param} and Theorem~\ref{thm:giant-diam}, we finally obtain our main result, namely:

\begin{theorem}\label{thm:helly-diam}
A diametral pair in a Helly graph can be computed w.h.p. in $\tilde{\cal O}(m\sqrt{n})$ time.
\end{theorem}

\begin{proof}
Let $k = 6\left\lceil\sqrt{n}\right\rceil$. By Corollary~\ref{cor:all-ecc} we can compute the set $B_k$ of all the vertices of eccentricity at most $k$ in ${\cal O}(mk) = {\cal O}(m\sqrt{n})$ time. There are now two cases. First assume that $B_k = V$. We can compute by dichotomic search the smallest $d \leq k$ such that $B_d = V$, which is exactly the diameter, in $\tilde{\cal O}(m\sqrt{n})$ time.
Then, a vertex is an end of a diametral path if and only if it is in $V \setminus B_{d-1}$, and by Corollary~\ref{cor:all-ecc} we can enumerate all such vertices in ${\cal O}(md) = {\cal O}(m\sqrt{n})$ time.
Otherwise, $B_k \neq V$, and so, $diam(G) > k$. Note that it implies $rad(G) > k/2 \geq 3\sqrt{n}$.
By Theorem~\ref{thm:giant-diam}, we can compute a diametral pair of $G$ w.h.p. in time $\tilde{\cal O}(mn/\sqrt{n}) = \tilde{\cal O}(m\sqrt{n})$.
\end{proof}

\section{Journey to the Center of $C_4$-free Helly graphs}\label{sec:helly-c4-free}

We now improve our results for the class of $C_4$-free Helly graphs.
Our results in this section are divided into three parts.
In Section~\ref{sec:central-vertex} we first explain how to compute a central vertex, and so the radius, in a $C_4$-free Helly graph.
We use this result and other properties in Section~\ref{sec:diametral-pairs}, in order to compute the diameter and a corresponding diametral pair. Finally, all the results in Section~\ref{sec:central-vertex} and Section~\ref{sec:diametral-pairs} are combined and enhanced in Section~\ref{sec:all-ecc} so as to compute the eccentricity of all vertices.

\subsection{Computing a central vertex}\label{sec:central-vertex}

We start with general properties of Helly graphs and $C_4$-free Helly graphs which we will then use in our analysis.
The first such property is a consequence of the unimodality of the eccentricity function in Helly graphs (see \cite{Dra89}). 

\begin{lemma}[\cite{Dra89}]\label{lem:intersect-center}
Let $G$ be a Helly graph. Then, for any vertex $v$ of $G$ and any farthest vertex $u \in F(v)$ we have $L(u,rad(G),v) \cap C(G) \neq \emptyset$.
\end{lemma}

Pseudo-modular graphs are exactly the graphs where each family of three pairwise intersecting balls has a common intersection \cite{BaM86}. Clearly, Helly graphs is a subclass of pseudo-modular graphs.

\begin{lemma}[\cite{BaM86}]\label{centroids}
   For every three vertices $x$, $y$, $z$ of a pseudo-modular graph $G$ there exist three shortest paths $P(x,y)$, $P(x,z)$,
   $P(y,z)$ connecting them such that either (1) there is a common vertex $v$ in $P(z,y) \cap P(x,z) \cap P(x,y)$
   or (2) there is a triangle $\bigtriangleup (x',y',z')$ in $G$ with edge $z'y'$ on $P(z,y)$,
   edge $x'z'$ on $P(x,z)$ and edge $x'y'$ on $P(x,y)$  (see Fig. \ref{fig:pseudomodular}). Furthermore, (1) is true if and only if $d(x,y)=p+q$, $d(x,z)=p+k$ and $d(y,z)=q+k$, for some $k,p,q\in \mathbb{N}$, and (2) is true if and only if $d(x,y)=p+q+1$, $d(x,z)=p+k+1$ and $d(y,z)=q+k+1$, for some $k,p,q\in \mathbb{N}$.
\end{lemma}
\begin{center}
 \centering\image{}{
   \includegraphics[scale=.55]{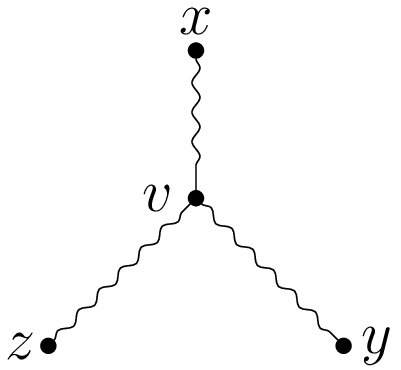}
 }\quad
 \image{}{
    \includegraphics[scale=.55]{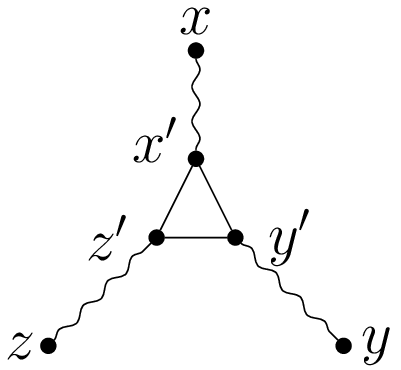}
 }
 \captionof{figure}{Vertices $x,y,z$ and three shortest paths connecting them in pseudo-modular graphs.}
 \label{fig:pseudomodular}
 \end{center}

The next properties are specific to $C_4$-free Helly graphs.
A set $S\subseteq V$ of a graph $G=(V,E)$ is called {\emph convex} if for every $x,y\in S$, $I(x,y)\subseteq S$ holds.

\begin{lemma}[\cite{Dra93}]\label{lm:ball-convex}  
Every ball of a $C_4$-free Helly graph is convex.
\end{lemma}

\begin{lemma}\label{lm:clique-slice}  For every vertices $v$ and $u$ of a $C_4$-free Helly graph $G$ and any integer $k\le dist(u,v)$, the set $L(u,k,v)$ is a clique.
\end{lemma}

\begin{proof} Consider any two vertices $x,y\in L(u,k,v)$ and assume that they are not adjacent. Let $\ell=dist(x,y)\ge 2$. Consider balls $N^1[x], N^{\ell-1}[y]$ and $N^{k-1}[u]$ in $G$. These balls pairwise intersect. By the Helly property, there must exist a vertex $z$ on a shortest path from $x$ to $y$ which is at distance at most $k-1$ from $u$. As by Lemma~\ref{lm:ball-convex} the ball $N^{dist(v,x)}[v]$ is convex, $z$ must belong to $N^{dist(v,x)}[v]$.
Thus, $dist(u,v)\leq dist(u,z)+dist(z,v)\leq k-1+dist(v,x)=dist(u,v)-1$, and a contradiction arises.
\end{proof}

We now introduce an important brick-basis of our approach.
The {\em multi-sweep} heuristic of Corneil et al. consists in performing a BFS~\cite{CDK03}, or a variant of it~\cite{CDHP01}, from an arbitrary vertex $v$, then from a farthest vertex $u \in F(v)$ (usually the last one visited), and finally to output $e(u)$ as an estimate of $diam(G)$. On general graphs, there may be an arbitrary gap between $diam(G)$ and the output of this heuristic~\cite{CDHP01}. However, on many graph classes it gives us a {\em constant additive approximation} of the diameter~\cite{CDEHV08,CDHP01,CDK03,Dra89}. We now prove that in particular, it is the case for $C_4$-free Helly graphs.

\begin{lemma}\label{lm:main}   Let $G$ be a $C_4$-free Helly graph with diameter $d$ and radius $r$. Let $s$ be an arbitrary vertex, $v$ be a vertex most distant from $s$, and $(x,y)$ be a diametral pair of $G$. Then, $e(v)\ge d-2$. 

Furthermore, if $e(v)= d-2$, then $e(v)=2r-3=dist(v,x)=dist(v,y)$ and $d=2r-1$. 
So, in particular, if $e(v)$ is even, then  $e(v)\ge d-1$. 
\end{lemma}

\begin{proof} By Lemma~\ref{lem:unimodal}, $d$ is either $2r$ or $2r-1$. Let $\ell=e(s)=dist(s,v)$. For vertices $s,v,x,y$ of $G$, we have $dist(x,y)=d$, $dist(s,x)\leq dist(s,v)=\ell$, $dist(s,y)\leq dist(s,v)=\ell$.
Furthermore, the three of $dist(v,x)$, $dist(v,y)$ and $dist(s,v)=\ell$ are at most $e(v)$.

First we show that, if $\max\{dist(v,y),dist(v,x)\}\le 2k$ for some integer $k$, then $d\leq 2k+1$.
By the triangular inequality, we have $2k \leq 2e(s) = 2\ell$.
Consider balls $N^{\ell-k}[s]$, $N^{k}[v]$, $N^{k}[y]$ in $G$. As $dist(v,y)\le 2k$ and $dist(s,y)\le \ell$, those balls pairwise intersect. By the Helly property, there is a vertex $a$ in $G$ belonging to all three balls. Necessarily, $dist(a,s)=\ell-k$, $dist(a,v)=k$ and $dist(a,y)\leq k$. Similarly, we can get a vertex $b$ in $G$ such that $dist(b,s)=\ell-k$, $dist(b,v)=k$ and $dist(b,x)\leq k$. As both $a$ and $b$ are in $L(v,k,s)$, by Lemma \ref{lm:clique-slice}, $dist(a,b)\le 1$. Thus, $d=dist(x,y)\leq dist(x,b)+dist(b,a)+dist(a,y)\leq 2k+1$.

Now, if $e(v)=2k$ for some integer $k$, then $\max\{dist(v,y),dist(v,x)\}\le e(v)=2k$ and, therefore,
$d\leq 2k+1=e(v)+1$. If $e(v)=2k+1$ for some integer $k$, then either $\max\{dist(v,y),dist(v,x)\}< e(v)=2k+1$ and hence $d\leq 2k+1=e(v)$ or
$\max\{dist(v,y),dist(v,x)\}= e(v)=2k+1$. As in the latter case $\max\{dist(v,y),dist(v,x)\}< 2k+2$, we also get $d\leq 2k+3=e(v)+2$.

\smallskip
In what follows, we consider this case, when $e(v)=2k+1=\max\{dist(v,y),dist(v,x)\}$, in more details. If $d$ is even (i.e., $d=2r$), then $d\le 2k+2$ and therefore $d\leq e(v)+1$. Assume now that $d$ is odd (i.e., $d=2r-1$) and that $d=e(v)+2=2k+3=2r-1$. That is, $r=k+2$. We will show that, under these conditions, $dist(v,y)=dist(v,x)=2k+1$ 
must hold.
For that assume w.l.o.g. that $dist(v,y)=\max\{dist(v,y),dist(v,x)\}=2k+1$. Since $v \in F(s)$, we have that $dist(s,y)\le \ell=dist(s,v)$. Furthermore, by the triangular inequality, we have $2k+1 \leq 2e(s) = 2\ell$, and so $\ell \geq k+1$. We shall use the following intermediate results:
\begin{itemize}
\item 
If $dist(s,y)= \ell$ then, by Lemma \ref{centroids}, there is a triangle $\bigtriangleup (v',s',y')$ in $G$ such that $dist(s,s')=\ell-k-1$, $dist(v',v)=dist(y,y')=k$.  Necessarily, $v'\in L(v,k,s)$ and $s'\in L(v,k+1,s)$.
\item
If $dist(s,y)\le \ell-1$, consider balls $N^{\ell-k-1}[s]$, $N^{k+1}[v]$, $N^{k}[y]$ in $G$. As these balls pairwise intersect, by the Helly property, there is a vertex $a$ in $G$ with $dist(a,s)=\ell-k-1$, $dist(a,v)=k+1$ and $dist(a,y)= k$. That is, $a\in L(v,k+1,s)$.
\item
If $dist(v,x)\le 2k$ then, as before, we can get a vertex $b$ in $G$ with $dist(b,s)=\ell-k$, $dist(b,v)=k$ and $dist(b,x)\leq k$. Necessarily, $b\in L(v,k,s)$.
\end{itemize}
Summarizing, we get the following combinations. If $dist(v,y)=2k+1$, $dist(v,x)\le 2k$ and $dist(s,y)= \ell$, then $d=dist(x,y)\leq dist(x,b)+dist(b,v')+dist(v',y')+dist(y',y)\le k+1+1+k=2k+2$ (notice that, by Lemma \ref{lm:clique-slice}, $dist(b,v')\le 1$), contradicting with  $d=2k+3$.
If $dist(v,y)=2k+1$, $dist(v,x)\le 2k$ and $dist(s,y)\le \ell-1$, then  $d=dist(x,y)\leq dist(x,b)+dist(b,a)+dist(a,y)\le k+2+k=2k+2$ (notice that, by Lemma \ref{lm:clique-slice}, $dist(b,a)\le 2$ as $a\in L(v,k+1,s)$ and $b\in L(v,k,s)$), contradicting with  $d=2k+3$.

Hence, $dist(v,x)=2k+1=dist(v,y)$ must hold. 
\end{proof}

We left open whether the lower-bound of Lemma~\ref{lm:main} can be refined to $e(v) \geq d-1$.
Note that this would be best possible.
Indeed although in some cases of interest, {\it e.g.} interval graphs, the output of the multi-sweep heuristic always {\em equals} the diameter~\cite{DNB97}, this nice property does not hold for strongly chordal graphs and so neither for $C_4$-free Helly graphs~\cite{CDHP01}.
Therefore, in Section~\ref{sec:diametral-pairs} we shall need additional tests in order to decide whether the output of this heuristic equals the diameter (and to compute a diametral pair when it is not the case).  

\smallskip
Before finally proving the main result of this subsection, we need the following gated property of Helly graphs.
The (weak) diameter of a set $S$ is equal to $diam(S) := \max_{x,y \in S} dist(x,y)$. 

\begin{lemma}\label{lem:gated}
Let $G$ be a Helly graph and $S$ be a subset of weak diameter at most two. Then, for any $v \notin S$ there exists a vertex $g_S(v) \in N^{dist(v,S)-1}[v] \cap \bigcap \{ N(x) \mid x \in Pr(v,S) \}$.
\end{lemma}

As it is standard~\cite{ChD94} we call such a vertex a {\em gate} of $v$, and we denote it by $g_S(v)$ -- we will omit the subscript if the set $S$ is clear from the context.

\begin{proof}
Since $S$ has weak diameter at most two the balls $N^{dist(v,S)-1}[v]$ and $N[x], \forall x \in Pr(v,S)$ pairwise intersect. Therefore, the result follows from the Helly property.
\end{proof}

\begin{remark}\label{rk:compute-gate}
For any fixed $S$ as above, we can compute a gate for every vertex $v \notin S$ in linear time.
 For that we first run a breadth-first search from $S$. Then, we recursively assign a gate to every vertex of $V \setminus S$, as follows. If $v \in N(S)$, then $v$ is its own gate and we set $p(v) = |N(v) \cap S|$. Otherwise, we choose for every vertex $v$ a father $u$, one step closer to $S$, that maximizes $p(u)$. Indeed, by induction, $p(u) = |Pr(u,S)|$. Then, we choose for $v$ the same gate as for its father $u$, and we set $p(v) = p(u)$. 

We observe that, more generally, if $S$ is an {\em arbitrary} vertex-subset (possibly, of weak diameter larger than two), then for every vertex $v$ with a gate in $N(S)$ this above procedure correctly computes such a gate. Indeed, to every vertex $v$ it associates a vertex $v^* \in N(S) \cap I(v,S)$ which maximizes $|N(v^*) \cap S|$. This observation is crucial in our proof of Theorem~\ref{thm:all-ecc}. 
\end{remark}

We are now ready to improve the result of Theorem~\ref{thm:helly-rad} for {\em $C_4$-free} Helly graphs, as follows:

\begin{theorem}\label{thm:compute-rad}
If $G$ is a $C_4$-free Helly graph then we can compute a central vertex and so $rad(G)$ in linear time.
\end{theorem}

\begin{proof}
Let $v$ be an arbitrary vertex, let $u \in F(v)$ and let $w \in F(u)$.
By Lemma~\ref{lm:main}, $e(u) = dist(u,w) \geq diam(G) - 2$.
Therefore, by Lemma~\ref{lem:unimodal}, $rad(G) \in \{ \left\lceil e(u)/2 \right\rceil, \left\lceil (e(u)+1)/2 \right\rceil, 1+ \left\lceil e(u)/2 \right\rceil\}$ (two of these numbers being equal, it gives us two possibilities).
In order to decide in which case we are, we use Lemma~\ref{lem:intersect-center}.
Indeed, if $rad(G) = r$ then $L(w,r,u) \cap C(G) \neq \emptyset$.
Furthermore, by Lemma~\ref{lm:clique-slice}, this set $C = L(w,r,u)$ is a clique.
We compute, for every $x \notin C$, its distance $dist(x,C)$ and a corresponding gate $g(x)$ -- which exists by Lemma~\ref{lem:gated}.
As observed in Remark~\ref{rk:compute-gate}, it takes linear time.
Then, $rad(G) = r$ implies $\max_{x\in V} dist(x,C) = r$.
If so then note that a vertex of $C$ has eccentricity $r$ if and only if it is adjacent to the gate of every vertex at a distance exactly $r$ from $C$.
Overall, in order to compute $rad(G)$ we pick the smallest $r$ such that a vertex of eccentricity $r$ can be extracted from $L(w,r,u)$.
\end{proof}

\subsection{Computing a diametral pair}\label{sec:diametral-pairs}

We base on the results from Section~\ref{sec:central-vertex} so as to prove the following theorem:

\begin{theorem}\label{thm:compute-diam}
If $G$ is a $C_4$-free Helly graph then we can compute a diametral pair and so $diam(G)$ in linear time.
\end{theorem}

\subsubsection*{Digression: an application to chordal Helly graphs}
Our results in the paper are proved valid assuming the input graph to be Helly.
However, the best-known recognition algorithms for this class of graphs run in quadratic time~\cite{LiS07}.
In what follows, we first explain an interesting application of Theorem~\ref{thm:compute-diam} to general chordal graphs.
We recall that it can be decided in linear time whether a given graph is chordal~\cite{RTL76}.

The Lexicographic Breadth-First-Search (LexBFS)~\cite{RTL76}, of which a description can be found in Fig.~\ref{fig:lexbfs}, is a standard algorithmic procedure that runs in linear time~\cite{HMPV00}.

\begin{figure}[h!]\centering
\includegraphics[width=.75\textwidth]{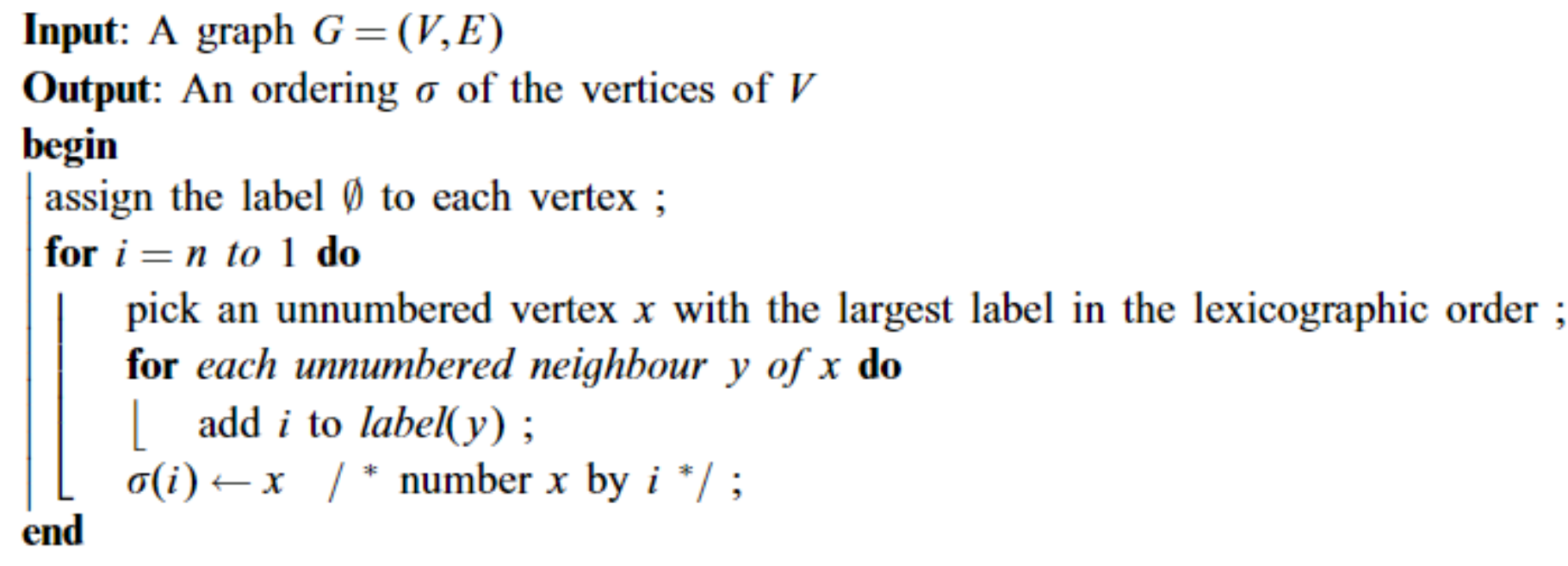}
\caption{Algorithm LexBFS~\cite{RTL76}.}
\label{fig:lexbfs}
\end{figure}

We use the following results on LexBFS in our analysis:

\begin{lemma}[\cite{DNB97}]\label{lem:additive-approx}
Let $v$ be the vertex visited last by an arbitrary LexBFS. If the graph is chordal, then the eccentricity of $v$ is within $1$ of the diameter.
\end{lemma}

\begin{lemma}[\cite{CDHP01}]\label{lem:odd-diam}
If the vertex $u$ of a chordal graph $G$ last visited by a LexBFS has odd eccentricity, then $e(u) = diam(G)$.
\end{lemma}

Altogether combined with Theorem~\ref{thm:compute-diam} we obtain that:

\begin{remark}\label{rk:check-result}
Consider an {\em arbitrary} chordal graph $G$.
If we assume $G$ to be Helly then, by Theorem~\ref{thm:compute-diam}, there exists a linear-time algorithm for computing a diametral pair of $G$.
Note that, we can apply this algorithm to $G$ without the knowledge that it is Helly, and either the algorithm will detect that $G$ is not Helly ({\it e.g.}, because some property of Helly graphs does not hold for $G$) or it will output some pair of vertices $(x,y)$.
Furthermore, if $G$ is chordal Helly, then $(x,y)$ is a diametral pair.
Let $d = dist(x,y)$.
We can check for a chordal graph $G$ whether $diam(G) = d$, or $G$ is not Helly, as follows:
\begin{itemize}
\item Let $u$ be the vertex visited last by a LexBFS. 
We may assume, by Lemma~\ref{lem:additive-approx}, that $e(u) \in \{d-1,d\}$ (otherwise, $d \neq diam(G)$, and so $G$ is not Helly) and, by Lemma~\ref{lem:odd-diam}, that $e(u)$ is even. Then, we compute $rad(G)$, which takes linear time~\cite{ChD94}. By Lemma~\ref{lem:additive-approx}, $diam(G) \in \{e(u),e(u)+1\}$, and so either $G$ is {\em not} Helly or, by Lemma~\ref{lem:unimodal}, $rad(G) \in \{ e(u)/2, e(u)/2 + 1 \}$.

\item If $e(u) = d$, $d$ is even and $rad(G) = d/2$, then this certifies that $diam(G) = e(u)$. Else, either $G$ is {\em not} Helly or we have $e(u) = d-1$, $d$ is odd and $rad(G) = e(u)/2 + 1$. Since $dist(x,y) = d = e(u) + 1$, we get $diam(G) = d$ by Lemma~\ref{lem:additive-approx}. 
\end{itemize}
\end{remark}

\subsubsection*{Proof of Theorem~\ref{thm:compute-diam}}

The remainder of this subsection is now devoted to the proof of Theorem~\ref{thm:compute-diam}.
For that, we first compute $r = rad(G)$, which by Theorem~\ref{thm:compute-rad} can be done in linear time.
We also apply the multi-sweep heuristic, {\it i.e.}, we pick an arbitrary vertex $v$ and we perform a BFS from a vertex $u \in F(v)$.
There are two main cases depending on the parity of $e(u)$.

\paragraph{Case $e(u)$ is even.} By Lemma~\ref{lm:main}, $e(u) \geq diam(G) - 1$. Since by Lemma~\ref{lem:unimodal} we have $diam(G) \geq 2r - 1$, it follows that $e(u) \in \{2r-2,2r\}$. Note that, in particular, if $e(u) = 2r$ then $diam(G) =e(u)$ and $u$ belongs to a diametral path. Otherwise, $diam(G) = e(u) + 1 = 2r-1$. We now explain how to compute a diametral pair in this latter subcase.

Let $w \in F(u)$. We may assume $e(w) = 2r-2$ (otherwise, $e(w) > dist(u,w) = e(u) = 2r-2$ and so, $w$ is an end of a diametral pair).
In this situation, $dist(u,w) = 2r-2$ and $u,w$ are mutually far apart.
The next result is a cornerstone of our algorithm:

\begin{lemma}\label{lem:projection}
Let $u,w$ be mutually far apart vertices in a $C_4$-free Helly graph $G$ such that $dist(u,w) = diam(G) - 1 = 2r-2$ is even, and let $C = L(u,r-1,w)$. Then, $(x,y)$ is a diametral pair of $G$ if and only if $dist(x,C) =dist(y,C) = r-1$ and $Pr(x,C) \cap Pr(y,C) = \emptyset$.
\end{lemma}

\begin{proof}
Since $e(u) = e(w) = 2r-2$, for any $x \in V$, the balls of radius $r-1$ and with centers $u,w,x$, respectively, pairwise intersect. The Helly property implies the existence of a vertex $c \in V$ such that $\max\{ dist(u,c), dist(w,c), dist(x,c) \} \leq r-1$. Since we also have $dist(u,w) = 2r-2$, we conclude that $c \in L(u,r-1,w) = C$ and $dist(x,C) \leq dist(x,c) \leq r-1$. 
Now on one direction, let $(x,y)$ be a diametral pair. 
By Lemma~\ref{lm:clique-slice}, $C$ is a clique, implying $dist(x,y) \leq dist(x,C) + 1 + dist(y,C) \leq 2r-1 = diam(G)$. Therefore, $dist(x,C) = dist(y,C) = r-1$.
For similar reasons, we must have $Pr(x,C) \cap Pr(y,C) = \emptyset$ (otherwise, $dist(x,y) \leq dist(x,C) + dist(y,C) \leq 2r-2 < diam(G)$, a contradiction).
Conversely, let $(x,y)$ be such that $dist(x,C) = dist(y,C) = r-1$ and $Pr(x,C) \cap Pr(y,C) = \emptyset$.
Suppose by contradiction $dist(x,y) < diam(G)$.
In particular, the balls of radius $r-1$ and respective centers $u,v,x,y$ pairwise intersect.
By the Helly property, there exists a $c \in C$ such that $\min\{dist(x,c),dist(y,c)\} \leq r-1 = dist(x,C) = dist(y,C)$.
But then, $c \in Pr(x,C) \cap Pr(y,C) = \emptyset$, a contradiction.
Hence, we proved that $(x,y)$ is a diametral pair.
\end{proof}

Our strategy now consists in computing a pair $(x,y)$ that satisfies the condition of this above Lemma~\ref{lem:projection}. We do so by using the ``gated property'' of Lemma~\ref{lem:gated}.
Indeed, let $C = L(u,r-1,w)$ be as above defined, and let $S = \{ x^* \mid \exists x \in V \ \text{such  that}  \ g_C(x) = x^* \land dist(x,C) = r-1\}$.
Since by Lemma~\ref{lm:clique-slice} $C$ is a clique, this set $S$ is well-defined and, according to Remark~\ref{rk:compute-gate}, it can be computed in linear time.
In order to compute a diametral pair of $G$, by Lemma~\ref{lem:projection} it is sufficient to compute a pair $x^*,y^* \in S$ such that $N(x^*) \cap N(y^*) \cap C = \emptyset$.
At first glance this approach does not look that promising since it is a particular case of the {\sc Disjoint Set} problem (sometimes called the monochromatic {\sc Orthogonal Vector}), that cannot be solved in truly subquadratic time under SETH~\cite{Wil05}.
Before presenting our solution to this special {\sc Disjoint Set} problem ({\it i.e.}, Lemma~\ref{lem:split}) we introduce an -- optional -- pre-processing so as to simplify a little bit the structure of our problem. 
For that we need the following lemma:

\begin{lemma}\label{lem:inclusion-proj}
In a $C_4$-free Helly graph $G$, for any clique $C$ and adjacent vertices $s,t \in N(C)$, the metric projections $Pr(s,C)$ and $Pr(t,C)$ are comparable, {\it i.e.,} either $Pr(s,C) \subseteq Pr(t,C)$ or $Pr(t,C) \subseteq Pr(s,C)$.
\end{lemma}

\begin{proof}
Let $s,t \in N(C)$ be  adjacent and suppose for the sake of contradiction that there exist $s^* \in Pr(s,C) \setminus Pr(t,C)$ and $t^* \in Pr(t,C) \setminus Pr(s,C)$. Then, $(s,t,t^*,s^*,s)$ induces a $C_4$.
\end{proof}

Let us compute $d^C(s) := |N(s) \cap C|$ for every $s \in S$.
It takes linear time.
We initialize $S_{red} := S$ and then we consider the vertices in $S$ sequentially.
At the time we consider a vertex $s \in S$, we check whether there exists a $t \in N(s) \cap S_{red}$ such that $d^C(t) \leq d^C(s)$. If it is the case then we remove $s$ from $S_{red}$.
Indeed, by Lemma~\ref{lem:inclusion-proj} it implies $N(t) \cap C = Pr(t,C) \subseteq Pr(s,C) = N(s) \cap C$. 
In particular, $N(s) \cap N(z) \cap C = \emptyset \Longrightarrow N(t) \cap N(z) \cap C = \emptyset$, and so we can safely discard vertex $s$.
Overall, the resulting subset $S_{red}$ is a stable set by construction.

A graph is split if its vertex-set can be bipartitioned in a clique and a stable set.
Note that by construction, the induced subgraph $H = G[C \cup S_{red}]$ is a split graph.
Computing the diameter of split graphs is already SETH-hard~\cite{BCH16}.
Fortunately, our split graph $H$ has some additional properties, namely we prove next that it is Helly.

\begin{lemma}\label{lem:clique-helly}
Let $u,w$ be two vertices in a $C_4$-free Helly graph $G$ such that $dist(u,w) = 2r-2$, let $C = L(u,r-1,v)$ and let $S \subseteq N(C)$ be a stable set.
Then, $H = G[C \cup S]$ is a split Helly graph.
\end{lemma}

\begin{proof}
By Lemma~\ref{lm:clique-slice}, the subset $C$ is a clique, hence $H$ is a split graph.
Furthermore, let us consider a family of pairwise intersecting balls in $H$.
We may assume that no such a ball is equal to $N_H[c], \forall c \in C$, or $N_H^2[s], \forall s \in S$, for all of these fully contain $C$.
In particular, there exists a subset $S' \subseteq S$ such that the subsets $N(z) \cap C, \ z \in S'$ pairwise intersect.
Then, we have that the balls of radius $r-1$ and with centers in $u,w$ and the balls of radius $1$ and with centers in the vertices of $S'$  pairwise intersect in $G$.
By the Helly property (applied to $G$), there exists a vertex $c$ at a distance $\leq r-1$ from both $u$ and $w$, and at a distance $\leq 1$ from all of $S'$.
Since $dist(u,w) = 2r-2$, we get $c \in L(u,r-1,w) = C$ and so, $c \in \bigcap \{ N(z) \cap C \mid z \in S' \}$.
Consequently, $H$ is Helly. 
\end{proof}

We are now left with computing a diametral pair for split Helly graphs.
 Actually since every split graph has constant diameter (at most three), then by Corollary~\ref{cor:all-ecc} the eccentricity of all vertices in a split Helly graph can be computed in total linear time. In what follows, we propose a different approach for computing the diameter of a split Helly graph than the one we presented in Corollary~\ref{cor:all-ecc}. Interestingly, this approach also works for other Helly-type properties, {\it e.g.} for split clique-Helly graphs and split open-neighbourhood-Helly graphs~\cite{LOS16}.

\begin{lemma}\label{lem:split}
A diametral pair in a split Helly graph can be computed in linear time.
\end{lemma}

\begin{proof}
Let $G=(C \cup S,E)$ be a split Helly graph with clique $C$ and stable set $S$ (note that if $C$ and $S$ are not given then they can be computed in linear time~\cite{Gol04}).
Assume $G$ to be connected and $diam(G) > 1$ (otherwise, we are done).
By the Helly property, $diam(G) = 2$ if and only if $G$ contains a universal vertex.
Furthermore, if it is the case then any pair $x,y$ of non-adjacent vertices is diametral.
Hence, from now on we assume that $diam(G) = 3$.
Let $G_0 := G$ and let $(x_1,x_2,\ldots,x_{|S|})$ be an arbitrary total order of $S$.
For every $1 \leq i \leq |S|$, we define $G_i := G[ \bigcap\{ N_G(x_j) \mid 1 \leq j \leq i \} \cup S ]$.
Our algorithm proceeds the vertices $x_i \in S$ sequentially, for $i = 1 \ldots |S|$, and does the following: If $x_i$ has eccentricity $3$ in $G_{i-1}$, then we compute a diametral pair in this subgraph which contains $x_i$ and we stop. 

We claim that our algorithm above is correct.
For that we prove by finite induction that for any $i \geq 0$, if the algorithm did not stop in less than $i$ steps then: \texttt{(i)} $G_i$ is connected; and \texttt{(ii)} $(x_p,x_q)$ is a diametral pair of $G_i$ if and only if it is a diametral pair of $G$. Since $G_0 = G$, this is true for $i=0$.
From now on we assume $i > 0$.
If the algorithm did not stop at step $i$ then (since in addition $G_{i-1}$ is connected by the induction hypothesis), $x_i$ has eccentricity two in $G_{i-1}$.
In particular, every vertex has a common neighbour with $x_i$, implying that there can be no isolated vertex in $G_{i}$.
We so obtain that $G_i$ is connected.
Furthermore, if $(x_p,x_q)$ is a diametral pair of $G$ then,   necessarily, it is also a diametral pair of the connected subgraph $G_i$ ({\it i.e.}, because $x_p$ and $x_q$ have no common neighbour in this subgraph, and so they are at distance $3$ to each other).
Conversely, let $(x_p,x_q)$ be a diametral pair of $G_i$.
Suppose, by contradiction, that $(x_p,x_q)$ is not a diametral pair of $G$, or equivalently $N_G(x_p) \cap N_G(x_q) \neq \emptyset$. 
Since the neighbour sets $N_G(x_j), \ j \in \{1,2,\ldots,i\} \cup \{p,q\}$ pairwise intersect, by the Helly property, there exists a vertex $w \in \bigcap\{ N_G(x_j) \mid j \in \{1,2,\ldots,i\} \cup \{p,q\} \}$.
But then, $(x_p,x_q)$ is not a diametral pair of $G_i$ (as $diam(G_i)\ge diam(G)=3)$, a contradiction.
As a result, our above algorithm for computing a diametral pair of $G$ is correct.

We still have to explain how to execute this algorithm in linear time.
For that, we maintain a partition of the clique, initialized to ${\cal P}_0 := (C)$.
At step $i$ we refine the former partition ${\cal P}_{i-1} = (C_1^{i-1},C_2^{i-1},\ldots,C^{i-1}_{k_{i-1}})$ into a new partition ${\cal P}_i = (C_1^{i-1} \cap N_G(x_i),C_1^{i-1}\setminus N_G(x_i),C_2^{i-1}\cap N_G(x_i),C_2^{i-1}\setminus N_G(x_i),\ldots,C^{i-1}_{k_{i-1}}\cap N_G(x_i),C^{i-1}_{k_{i-1}}\setminus N_G(x_i))$.
This partition refinement can be done in time ${\cal O}(N_G(x_i))$ (up to some initial pre-processing in ${\cal O}(|C|)$ time)~\cite{HMPV00}.
Furthermore, an easy induction proves for any $i \geq 0$ that the first group of ${\cal P}_i$ is exactly $C_1^i = \bigcap\{ N_G(x_j) \mid 1 \leq j \leq i \}$ {\it i.e.}, the clique of $G_i$.
We finally explain how we use this partition in order to decide, at step $i$, whether $x_i$ has eccentricity equal to $3$ in $G_{i-1}$.
At the beginning of the algorithm we compute the degree of every vertex in $S$.
Then, at step $i$ we consider all the vertices in $C_2^i = \bigcap\{ N_G(x_j) \mid 1 \leq j < i \} \setminus N_G(x_i)$ sequentially (second group of the partition ${\cal P}_i$).
For every $w \in C_2^i$ we enumerate all its neighbours in $S$ and we decrease their respective degrees by one.
In particular, if during step $i$ the degree of some vertex $x_p \in S$ falls to $0$, then $x_p$ has no common neighbour with $x_i$ in $G_{i-1}$.
Equivalently, $(x_i,x_p)$ is a diametral pair of $G_{i-1}$ and the eccentricity of $x_i,x_p$ in this subgraph is $3$.
We observe that the sets $C_2^i$ on which we iterate are pairwise disjoint.
As a result, the total complexity of the algorithm is linear.
\end{proof}

This above Lemma~\ref{lem:split} achieves proving Theorem~\ref{thm:compute-diam} in the case when $e(u)$ is even.

\paragraph{Case $e(u)$ is odd.} By Lemma~\ref{lm:main}, $e(u) \geq diam(G) - 2$.
Therefore, by Lemma~\ref{lem:unimodal},  $e(u) \in \{2r-3,2r-1\}$.
In the first subcase, we deduce from Lemma~\ref{lem:unimodal} that $diam(G) = 2r-1$.
Furthermore, we can compute a diametral pair as follows.
Let $w \in F(u)$, and assume $e(w) = e(u) = dist(u,w) = 2r-3$ (otherwise, either $e(w) = 2r-2$ is even and we are back to the former case, or $e(w) = 2r-1$ and then we are done since $w$ is an end of a diametral path).

\begin{lemma}\label{lem:projection-odd}
Let $u,w$ be mutually far apart vertices in a $C_4$-free Helly graph $G$ such that $dist(u,w) = diam(G) - 2 = 2r-3$ is odd, and let $A = L(w,r-2,u)$. Then, $(x,y)$ is a diametral pair of $G$ if and only if $dist(x,A) =dist(y,A) = r-1$, $Pr(x,A) \cap Pr(y,A) = \emptyset$, and in addition $dist(y,u)=dist(y,w)=dist(x,u)=dist(x,w)=2r-3$.
\end{lemma}

\begin{proof}
We first prove that for every vertex $x$ we have $dist(x,A) \leq r-1$.
Indeed, since $w,u$ are mutually far apart, we have $dist(x,u)\le dist(w,u)=2r-3$ and $dist(x,w)\le dist(w,u)=2r-3$. Hence, balls $N^{r-1}[u]$,  $N^{r-2}[w]$ and  $N^{r-1}[x]$ pairwise intersect. By the Helly property, there is a vertex $z\in A$ with $dist(x,z)\le r-1$.

For any pair $(x,y)$ we have $dist(x,y) \leq dist(x,A) + dist(Pr(x,A),Pr(y,A)) + dist(y,A) \leq 2r-2 + dist(Pr(x,A),Pr(y,A))$.
By Lemma~\ref{lm:clique-slice}, $A$ is a clique, which implies $dist(Pr(x,A),Pr(y,A)) \leq 1$.
As a result, if $(x,y)$ is diametral, we get $dist(x,A) = dist(y,A) = r-1$ and $Pr(x,A) \cap Pr(y,A) = \emptyset$.  
We also get $dist(y,u)=dist(y,w)=dist(x,u)=dist(x,w)=2r-3$ by Lemma~\ref{lm:main}.
Conversely let $(x,y)$ be any pair that satisfies all these above properties, and suppose for the sake of contradiction that we have $dist(x,y) \leq 2r-2$.
Consider balls $N^{r-1}[x], N^{r-1}[y], N^{r-1}[u]$ and $N^{r-2}[w]$. By distance requirements, these balls pairwise intersect. The Helly property implies that a vertex $z\in A$ exists such that
$dist(x,z)= r-1$ and $dist(y,z)= r-1$. However, the latter contradicts with $Pr(x,A)\cap Pr(y,A)=\emptyset$. 
\end{proof}

We observe that Lemma~\ref{lem:projection-odd} is quite the same as Lemma~\ref{lem:projection} and that the same techniques can be used in order to compute a diametral pair in this subcase.

\smallskip
The most difficult subcase is when $e(u) = 2r-1$. 
By Lemma~\ref{lem:unimodal}, either $diam(G) = 2r-1$ or $diam(G) = 2r$.
We explain below how, assuming $diam(G) = 2r$, we can compute in linear time {\em all} central vertices. 
Then, if $diam(G) = 2r$, by Lemma~\ref{lem:intersect-center}, a pair $(x,y)$ is diametral if and only if both $x$ and $y$ are at a distance exactly $r$ from all central vertices. 
In particular, we can pick any such a pair $(x,y)$ and check whether we have $dist(x,y) = 2r$ (otherwise, $diam(G) = 2r-1$ and $u$ is an end of a diametral path).

Assume w.l.o.g. $r \geq 2$ (otherwise, computing a diametral pair is trivial).

\begin{lemma}\label{lem:clique-center}
If $G$ is a $C_4$-free Helly graph of radius $r \geq 2$ and diameter $2r$ then $C(G)$ is a clique. 
\end{lemma}

\begin{proof}
Every central vertex is at a distance exactly $r$ from both ends $x,y$ of any diametral path.
In particular, $C(G) \subseteq L(x,r,y)$, which is a clique by Lemma~\ref{lm:clique-slice}.
\end{proof}

Therefore, if $diam(G) = 2r$, by Lemma~\ref{lem:clique-center}, for any central vertex $c$, we have $C(G) \subseteq N[c] = S$.
Note that, by Theorem~\ref{thm:compute-rad}, we can compute such a central vertex $c$ in linear time.
Furthermore, $diam(S) \leq 2$.
For every $x \notin S$, we compute $dist(x,S)$ and a corresponding gate $g(x)$, which exists by Lemma~\ref{lem:gated}.
Note that, according to Remark~\ref{rk:compute-gate}, it takes linear time.

By construction, $\max_x dist(x,S) \leq dist(x,c) \leq r$.
Furthermore, every vertex at a distance $\leq r - 2$ from $S$ is at a distance $\leq r$ from every vertex of $S$.
As a result, we only need to consider the vertices at a distance $\geq r - 1$ from $S$.
In fact, and as already observed in the proof of Theorem~\ref{thm:compute-rad}, for a vertex of $S$ to be central it needs to be adjacent to the gates of all the vertices at a distance exactly $r$ from $S$. All the vertices which satisfy this necessary condition can be computed in linear time.
Hence, we can restrict ourselves to the vertices that are at distance exactly $r-1$ from $S$.

\begin{lemma}
Let $G$ be a $C_4$-free Helly graph and let $S$ be such that $diam(S) \leq 2$ and $\max_x dist(x,S) = r$.
If $dist(x,S) = r -1$, then there exists a vertex $pg_S(x) \in N^{r-1}[x]$ such that $S \cap N^{r}[x] \subseteq N[pg_S(x)]$.
Moreover, $pg_S(x)$ is in the closed neighbourhood of some gate of $x$.
\end{lemma}

We call such a vertex a {\em pseudo-gate} of $x$.

\begin{proof}
The existence of such a pseudo-gate follows from the fact that the balls $N^{r-1}[x]$ and $N[s], \ s \in S \cap N^{r}[x]$ pairwise intersect, and from the Helly property.
Now, let $pg(x)$ and $g(x)$ be arbitrary pseudo-gate and gate of $x$, respectively, and assume $pg(x) \neq g(x)$ and $pg(x)g(x) \notin E$ (else, we are done).
In particular, $pg(x) \notin Pr(x,S)$.
Note that, since we have $Pr(x,S) \subseteq N(g(x)) \cap N(pg(x))$, $dist(g(x),pg(x)) = 2$.
Then, the balls $N^{r-2}[x], N[g(x)] \ \text{and} \ N[pg(x)]$ pairwise intersect.
By the Helly property, there exists a vertex $x^*$ in their common intersection.
We claim that $x^*$ is a gate of $x$, that will prove the lemma.
Indeed, for every $s \in Pr(x,S)$ we get a cycle $(s,pg(x),x^*,g(x),s)$.
Since $G$ is $C_4$-free, this implies $sx^* \in E$.
\end{proof}

\begin{remark}\label{rk:compute-pseudo-gate}
For every $x \in N(S)$, we can choose as its pseudo-gate any vertex of $N[x]$ that maximizes the intersection of $S$ with its closed neighbourhood (possibly, $x$ itself).
Then, when we compute a gate for every vertex, we break ties by choosing one such a gate whose pseudo-gate maximizes its intersection with $S$.
In doing so, we can compute a pseudo-gate for every vertex at distance $r-1$ from $S$, in total linear time.
\end{remark}

Altogether combined if $diam(G) = 2r$ then the central vertices of $S$ are exactly those that are adjacent to the gates of all the vertices at distance exactly $r$ from $S$ {\em and} either equal or adjacent to the pseudo-gates of all the vertices at distance exactly $r-1$ from $S$.

\subsection{Computing all eccentricities}\label{sec:all-ecc}

We are now ready to present the main result of this section, namely:

\begin{theorem}\label{thm:all-ecc}
If $G$ is a $C_4$-free Helly graph then we can compute the eccentricity of all vertices in linear time.
\end{theorem}

The remainder of this subsection is devoted to the proof of Theorem~\ref{thm:all-ecc}, or equivalently, by Lemma~\ref{lem:ecc-formula}, how to compute $C(G)$ in linear time for a $C_4$-free Helly graph $G$.
 Our main tool for that is our parameterized linear-time algorithm for {\sc Small Eccentricities}, for $k=2$ (Theorem~\ref{thm:all-ecc-param}). In particular, by Corollary~\ref{cor:all-ecc}, we may assume that $rad(G) = r \geq 3$.
Let $(x,y)$ be a diametral pair.
By Theorem~\ref{thm:compute-diam}, it can be computed in linear time.
There are two cases, depending on the parity of $dist(x,y) = d = diam(G)$.

If $d = 2r$ is even then, by Lemma~\ref{lem:clique-center}, $C(G) \subseteq L(x,r,y)$ is a clique.
We can reuse the same idea as for Theorem~\ref{thm:compute-rad} in order to extract all the central vertices from $L(x,r,y)$ in linear time (see also Section~\ref{sec:diametral-pairs}, Subcase $e(u) = 2r-1$, for a more complicated method which does not need the pre-computation of a diametral pair).

From now on we assume $d = 2r-1$ is odd.
Obviously, $C(G) \subseteq N^r[x] \cap N^r[y]$.
Since $dist(x,y) = 2r-1$ we get $N^r[x] \cap N^r[y] = L(x,r-1,y) \cup L(y,r-1,x) \cup Z$, where $Z := \{ z \in V \mid dist(z,x) = dist(z,y) = r \}$.
Furthermore, for a $C_4$-free Helly graph, by Lemma~\ref{lm:clique-slice} the slices $L(x,r-1,y)$ and $L(y,r-1,x)$ are cliques; again we can reuse the same idea as for Theorem~\ref{thm:compute-rad} in order to extract all the central vertices in these two disjoint sets in linear time.
Note that, by Lemma~\ref{lem:intersect-center}, there must be at least one such central vertex in both cliques. 

From now on we focus on $Z$.
\begin{myclaim}\label{claim:two-meighbours}
Every vertex $z \in Z$ has two adjacent neighbours in $L(x,r-1,y)$ and $L(y,r-1,x)$, respectively.
\end{myclaim}
\begin{proofclaim}
Let $p = q = r-1$ and $k = 0$. We have $dist(x,y) = p+q+1, \ dist(x,z) = p+k+1, \ dist(y,z) = q + k+1$.
Since $G$ is pseudo-modular, by Lemma~\ref{centroids}, $z$ is adjacent to the two ends of an edge in the middle of some shortest $xy$-path.
\end{proofclaim}
This above claim has many important consequences.
The first one is that, since both $L(x,r-1,y)$ and $L(y,r-1,x)$  are cliques, $Z$ has weak diameter at most $3 \leq r$. 
Furthermore, still by Claim~\ref{claim:two-meighbours}, every vertex of $Z$ is at a distance $\leq 2 < r$ from all the vertices of $L(x,r-1,y) \cup L(y,r-1,x)$. 
As a result, in order to decide whether a vertex of $Z$ is central, we only need to consider its distances to the vertices {\em outside} of $N^r[x] \cap N^r[y]$.
Note that in particular, every vertex at a distance $\leq r-2$ from $L(x,r-1,y) \cup L(y,r-1,x)$ is at a distance $\leq r$ from every vertex of $Z$.
Therefore, we further restrict our study to the vertices $w$ such that $dist(w,L(x,r-1,y) \cup L(y,r-1,x)) \geq r-1$. 
Recall that, by Lemma~\ref{lem:intersect-center}, we have $L(x,r-1,y) \cap C(G) \neq \emptyset$ and similarly $L(y,r-1,x) \cap C(G) \neq \emptyset$. 
Hence, $dist(w,L(x,r-1,y) \cup L(y,r-1,x)) \in \{r-1,r\}$.

\paragraph{Subcase $dist(w,L(x,r-1,y) \cup L(y,r-1,x)) = r-1$.}
Let $C := L(x,r-1,y)$ and consider the set $W$ of all vertices $w$ such that $dist(w,L(x,r-1,y) \cup L(y,r-1,x)) = dist(w,C) = r-1$ (the subcase when $dist(w,L(y,r-1,x)) = r-1$ is symmetric to this one). Let $A$ contain a gate $g_C(w)$ for every $w \in W$, which exists by Lemma~\ref{lem:gated}. Recall that the set $A$ can be computed in linear time (see Remark~\ref{rk:compute-gate}).
\begin{myclaim}\label{claim:reduce-hs}
 For every $z \in Z$, we have $\max_{w \in W} dist(z,w) \leq r$ if and only if $A \subseteq N^2[z]$.
\end{myclaim}
\begin{proofclaim}
First assume that $\forall a \in A, \ N[a] \cap N[z] \neq \emptyset$. 
Let $w \in W$ be arbitrary and let $a := g_C(w) \in A$.
Then, $dist(z,w) \leq dist(z,a) + dist(a,w) \leq 2 + (r-2) = r$.
Conversely, let us assume that we have $\max_{w \in W} dist(z,w) \leq r$, and let $w \in W$ be arbitrary.
The balls $N^{r-1}[x], N^{r-1}[w], N^r[y], N[z]$ pairwise intersect. 
Therefore, by the Helly property, $N[z] \cap Pr(w, C) \neq \emptyset$.
For any gate $g_C(w) \in A$, it implies $N[z] \cap N[g_C(w)] \supseteq 
N[z] \cap Pr(w, C) \neq \emptyset$. 
\end{proofclaim}
 Overall with this above claim we are reduced to {\sc Small Eccenttricities}, with $k=2$, which by Theorem~\ref{thm:all-ecc-param} can be solved in linear time.

\paragraph{Subcase $dist(w,L(x,r-1,y) \cup L(y,r-1,x)) = r$.}
Let $W$ contain all vertices $w$ such that $dist(w,L(x,r-1,y) \cup L(y,r-1,x)) = r$.
We will need the following properties of this subset $W$.
\begin{myclaim}
For every $w \in W$ we have $dist(x,w) = dist(y,w) = 2r-1$ and $dist(w,Z) = r-1$.
Moreover, $Pr(w,Z) = L(w,r-1,x) \cap L(w,r-1,y)$ is a clique.
\end{myclaim}
\begin{proofclaim}
We first prove that $dist(x,w) = dist(y,w) = 2r-1$. Indeed, suppose for the sake of contradiction that $dist(x,w) \leq 2r-2$. Since the balls $N^{r-1}[x], N^{r-1}[w], N^r[y]$ pairwise intersect, by the Helly property, we get $dist(w,L(x,r-1,y)) \leq r-1$, a contradiction.
Hence, the balls $N^{r-1}[w], N^r[x], N^r[y]$ pairwise intersect, which implies, by the Helly property, $dist(w,Z) = r-1$.
In this situation, we have $Pr(w,Z) = L(w,r-1,x) \cap L(w,r-1,y)$.
By Lemma~\ref{lm:clique-slice}, $Pr(w,Z)$ is a clique.
\end{proofclaim}
The combination of this above claim with Lemma~\ref{lem:gated} implies the existence of a gate $g_Z(w) \in N^{r-2}[w] \cap \bigcap \{ N(v) \mid v \in Pr(w,Z) \}$ for every vertex $w \in W$.
Note that we can compute the gate of all such vertices $w$ in linear time, by using our classical BFS method in order to compute all the gates (see Remark~\ref{rk:compute-gate}).
-- Incidentally, this algorithm will also associate a vertex $g_Z(v)$ to every $v \notin W$, but the latter may not be a gate of $v$. -- 
So, let $A$ contain a gate $g_Z(w)$ for every vertex $w \in W$.
We prove as before:
\begin{myclaim}
 For every $z \in Z$, we have $\max_{w \in W} dist(z,w) \leq r$ if and only if $A \subseteq N^2[z]$.
\end{myclaim}
\begin{proofclaim}
We can prove this above condition to be sufficient for having $\max_{w \in W} dist(z,w) \leq r$ in the exact same way as we did for Claim~\ref{claim:reduce-hs}. 
Conversely, let us assume that we have $\max_{w \in W} dist(z,w) \leq r$, and let $w \in W$ be arbitrary.
The balls $N[z], N^{r-1}[w], N^r[x], N^r[y]$ pairwise intersect, and so, by the Helly property, $N[z] \cap Pr(w,Z) = N[z] \cap (L(w,r-1,x) \cap L(w,r-1,y)) \neq \emptyset$.
It implies that $N[z] \cap N[g_Z(w)] 
\supseteq
N[z] \cap Pr(w,Z) \neq \emptyset$.
\end{proofclaim}
 We are done by reducing a final time to {\sc Small Eccentricities}, with $k=2$, which by Theorem~\ref{thm:all-ecc-param} can be solved in linear time.

\section{More reductions to split graphs}\label{sec:diam-chordal}

We conclude by considering diameter computation within another class  than Helly graphs, namely chordal graphs.
This is motivated by our results in Section~\ref{sec:diametral-pairs} where we reduced the problem of computing a diametral pair on chordal Helly graphs to the same problem on split Helly graphs.
We prove next that there exists a (randomized) reduction from diameter computation on {\em general} chordal graphs to the same problem on split graphs.

The sparse representation of a split graph is the list of the closed neighbourhoods of vertices in its stable set~\cite{DHV19+b}.
The {\sc Disjoint Set} problem consists in computing the diameter of a split graph given by its sparse representation.
For a split graph $H$ with stable set $U$ we define $\ell(H) := \sum_{u \in U} deg_H(u)$, {\it a.k.a.} the size of its sparse representation.

\begin{theorem}\label{thm:chordal-reduction}
For any chordal graph $G$, we can compute in linear time the sparse representations of a family $(H_i)$ of split graphs such that:
\begin{itemize}
\item If for every $i$ we can compute $diam(H_i)$ in time ${\cal O}(|E(H_i)|^a \cdot |V(H_i)|)$, then we can compute $diam(G)$ in time $\tilde{\cal O}(m^{1+a})$;
\item If for every $i$ we can compute $diam(H_i)$ in time ${\cal O}(\ell(H_i) \cdot |V(H_i)|^b)$, then we can compute $diam(G)$ in time $\tilde{\cal O}(mn^b)$.
\end{itemize}
\end{theorem}

 An interesting byproduct of our reduction, proved in Section~\ref{sec:ecc-chordal}, is that we can approximate in quasi linear time the eccentricity of all vertices in a chordal graph with a one-sided additive error of at most one. This answers an open question from~\cite{Dra19}. Finally, in Section~\ref{sec:vc-dim} we give another application of Theorem~\ref{thm:chordal-reduction} to chordal graphs of constant VC-dimension.

\subsection{Preliminaries}
We shall use the following metric properties of chordal graphs.
We stress that these are quite similar to some metric properties of $C_4$-free Helly graphs that we proved in Section~\ref{sec:helly-c4-free}.

\begin{lemma}[\cite{ChD94}]\label{lem:chordal-gate}
For any clique $C$ in a chordal graph $G$, we can compute in linear time the distance $dist(x,C)$ and a gate $g(x) \in \bigcap \{ I(x,w) \mid w \in Pr(x,C) \}$ adjacent to all vertices from $Pr(x,C)$, for all vertices $x$.
\end{lemma}

\begin{lemma}[\cite{ChD94}]\label{lem:chordal-inclusion-proj}
In a chordal graph $G$, for any clique $C$ and adjacent vertices $s,t \notin C$ the metric projections $Pr(s,C)$ and $Pr(t,C)$ are comparable, {\it i.e.}, either $Pr(s,C) \subseteq Pr(t,C)$ or $Pr(t,C) \subseteq Pr(s,C)$.
\end{lemma}

\subsection{The reduction}\label{sec:chordal-reduction}
Our reduction is one-to-many.
We recall that a clique-tree of a graph $G$ is a tree $T$ of which the nodes are the maximal cliques of $G$, and such that for every vertex $v$ the set of all the maximal cliques that contain $v$ induces a connected subtree.
It is known~\cite{BlP93} that $G$ is chordal if and only if it has a clique-tree and, furthermore, a clique-tree can be computed in linear time.

We may see a clique-tree $T$ as a node-weighted tree where, for any maximal clique $C$, $w(C) = |C|$.
Then, let $w(T) := \sum_C w(C)$.
For a chordal graph, $w(T) = {\cal O}(n+m)$~\cite{BlP93}.
We will use a standard result on weighted centroids in trees, namely:

\begin{lemma}[\cite{Gol71}]\label{lem:weighted-centroid}
Every node-weighted tree $T$ has at least one {\it weighted centroid}, that is, a node $v$ whose removal leaves components of maximum weight $\leq w(T)/2$. Moreover, a weighted centroid can be computed in linear time.  
\end{lemma}

Let $T$ be a fixed clique-tree of $G$.
If $T$ is reduced to a single node, or to exactly two nodes, respectively (equivalently, either $G$ is a complete graph or it is the union of two crossing complete subgraphs), then we output $diam(G) = 1$, or $diam(G) = 2$, respectively (base case of our reduction).
Otherwise, let $S \in V(T)$ be a weighted centroid of $T$. 
By Lemma~\ref{lem:weighted-centroid}, the clique $S$ can be computed in time ${\cal O}(n)$ if we are given $T$ in advance. 
Furthermore, let $T_1,T_2,\ldots,T_{\ell}$ be the components of $T \setminus \{S\}$, and for every $1 \leq i \leq \ell$ let $V_i := \left(\bigcup V(T_i)\right) \setminus S$.
It is known~\cite{BlP93} that the sets $V_1,V_2,\ldots,V_{\ell}$ are exactly the connected components of $G \setminus S$.
Since $\forall i, N_G(V_i) \subseteq S$ is a clique, the closed neighbourhoods $N_G[V_i], 1 \leq i \leq \ell$ induce distance-preserving subgraphs of $G$, which we denote by $G_1,G_2,\ldots, G_{\ell}$.
We apply our reduction recursively on each of these subgraphs $G_i$.
Then, let $d_S := \max_{v_i \in V_i, v_j \in V_j \mid i \neq j} dist(v_i,v_j)$.
We have: $$diam(G) = \max \{d_S, \max\{ diam(G_i) \mid 1 \leq i \leq \ell\}\}.$$

We are left with computing $d_S$.
For that, we define $\forall i, d_i := \max_{v_i \in V_i} dist(v_i,S)$.
We order the sets $V_i$ by non-increasing value of $d_i$.
Since $S$ is a clique, we get $d_1 + d_2 \leq d_S \leq d_1 + d_2 + 1$.
In order to decide in which case we are, we proceed as follows:
\begin{itemize}
\item We discard all sets $V_i$ such that  $d_i < d_2$. Doing so we are left with sets $V_1,V_2,\ldots,V_k, \ k \leq \ell$.
\item Then, for every $1 \leq i \leq k$ and $v \in V_i$, if $dist(v,S) = d_i$ then we compute a gate for vertex $v$, which exists by Lemma~\ref{lem:chordal-gate}.
Furthermore, if two such gate vertices are adjacent then they must be in the same connected component of $G \setminus S$, and by Lemma~\ref{lem:chordal-inclusion-proj} their respective metric projections on $S$ are comparable.
It implies that we can remove any of these two vertices with largest metric projection on $S$ (see Section~\ref{sec:diametral-pairs} for a similar idea on $C_4$-free Helly graphs).
Thus, from now on, we assume all selected gate vertices to be pairwise non-adjacent.
\item Finally, let $a,b \notin V$ be fresh new vertices which we make  adjacent to each other and to all vertices of $S$. There are two subcases:
\begin{itemize}
\item Case $d_1 \neq d_2$. We make vertex $a$ adjacent to all gates of the vertices in $V_1$, while we make vertex $b$ adjacent to all gates of the vertices in $\bigcup_{i=2}^k V_i$.
\item Case $d_1 = d_2$. For every $1 \leq i \leq k$, with probability $1/2$ we make vertex $a$ adjacent to all gates of the vertices in $V_i$ (otherwise, we do so with vertex $b$).
\end{itemize}
\end{itemize}
Doing as above we get a split graph $H$ whose clique and stable set are $S \cup \{a,b\}$ and the selected gate vertices, respectively. 

\begin{myclaim}
If $diam(H) = 3$ then $d_S = d_1 + d_2 + 1$.
Conversely, if $d_S = d_1 + d_2 + 1$ then $diam(H) = 3$ with probability $\geq 1/2$.
\end{myclaim}

\begin{proofclaim}
On one direction, let $x^* \ \text{and} \ y^*$ be two gate vertices such that $dist_H(x^*,y^*) = 3$. By construction of $H$, we have that $x^*,y^*$ are the respective gates of two vertices $x,y$ such that $\min \{dist_G(x,S),dist_G(y,S\} \geq d_2$. Furthermore, since $\{a,b\} \cap (N_H(x^*) \cap N_H(y^*)) = \emptyset$, we get that $x,y$ are in different components of $G \setminus S$, and $\{dist_G(x,S),dist_G(y,S\} = \{d_1,d_2\}$. Hence, $d_S \geq dist(x,y) \geq d_1 + d_2 + 1$.
On the other direction, let us assume the existence of a pair $(x,y)$ such that: $x$ and $y$ are in different connected components of $G \setminus S$, and $dist(x,y) = d_1 + d_2 + 1$. Without loss of generality, let $dist(x,S) = d_1$ and $dist(y,S) = d_2$. Let also $x^*,y^*$ be the two gates computed for $x$ and $y$, respectively (we may assume, without loss of generality, that $x^*,y^*$ are indeed in the stable set of $H$). We must have $N_H(x^*) \cap N_H(y^*) \subseteq \{a,b\}$. In particular,  if $d_1 \neq d_2$ then $N_H(x^*) \cap N_H(y^*) = \emptyset$, else $Prob[  N_H(x^*) \cap N_H(y^*) = \emptyset ] \geq 1/2$.
\end{proofclaim}

Overall, we may repeat the construction of this above split graph $H$ up to ${\cal O}(\log{n})$ times in order to compute $d_S$ with high probability.

\subsection{Analysis}\label{sec:chordal-analysis}

Since at every step of our reduction we pick a weighted centroid in the clique-tree of every subgraph $G_i$ considered, there are ${\cal O}(\log{w(T)}) = {\cal O}(\log{n})$ recursion levels.
Therefore, up to polylogarithmic factors, the total running-time of the reduction is of the same order of magnitude as the worst-case running time of a single step.
Furthermore, it is not hard to prove that the first step, when we only consider the full input graph $G$, runs in linear time ({\it i.e.}, omitting the computation of the diameter for the related split graph $H$).
However, during the next steps of our reduction we may need to consider pairwise overlapping subgraphs $G_i$, thereby making the analysis more delicate. 

The key insight here is that the clique-trees of all these subgraphs form a family of pairwise disjoint subtrees of $T$.
We next explain how to perform the first step, and so all subsequent ones, in time ${\cal O}(w(T))$.
Since $w(T) = {\cal O}(n+m)$~\cite{BlP93}, doing so we can compute all the desired split graphs $H$ throughout our reduction in total quasi linear time $\tilde{\cal O}(n+m)$.

\begin{lemma}\label{lem:gate-with-tree}
Let $S$ be any clique of a chordal graph $G$. If a clique-tree $T$ is given, then in time ${\cal O}(w(T))$ we can compute $\forall v \notin S, dist(v,S)$ and a corresponding gate $g(v)$.
\end{lemma}

\begin{proof}
Let ${\cal C} := V(T) \cup \{S\}$ be the set of all maximal cliques of $G$, to which we also add the clique $S$ if it is not maximal.
We define a set of fresh new vertices indexed by ${\cal C}$, namely let $X_{\cal C} := \{ x_C \mid C \in {\cal C} \}$.
Then, let $J := (V \cup X_{\cal C}, \{ vx_C \mid v \in C \})$ be the vertex-clique incidence graph of $G$.
Note that we can construct $J$ by scanning once the clique $S$ and all the maximal cliques of $G$.

We prove as a subclaim that for every vertex $v$ we have $dist_{J}(v,x_S) = 2 \cdot dist_G(v,S) + 1$.
Indeed, since by construction $N_J(x_S) = S$, we get $dist_{J}(v,x_S) = 1 + \min_{u \in S} dist_J(v,u) = 1 + dist_J(v,S)$. 
Furthermore, in every $vS$-path of $J$, that is in every path between $v$ and a closest vertex of $S$, half of the internal vertices must be in $X_{\cal C}$.
Since in addition two vertices that are adjacent to a same maximal clique in $J$ are adjacent in $G$, it allows us to transform such a path in $J$ to a $vS$-path in $G$ that is twice shorter.
Conversely, any $vS$-path of $G$ can be transformed into a $vS$-path of $J$ that is twice longer, simply by adding between every two consecutive vertices a maximal clique which contains both of them.
-- Note that combining the two constructions, the latter exactly characterizes the shortest $vS$-paths in $J$. --
As a result, we proved as claimed that $dist_J(v,S) = 2 \cdot dist_G(v,S)$.
It implies that after a BFS in $J$ rooted at $x_S$ we get $dist_G(v,S), \forall v \notin S$.

Then, we recursively define $p(\alpha), \forall \alpha \in V(J) \setminus N_J[x_S]$ as follows (recall that $N_J[x_S] = \{x_S\} \cup S)$:
if $\alpha = x_C$ for some $C \in {\cal C} \setminus \{S\}$, and $C \cap S \neq \emptyset$ (equivalently, $dist_{J}(x_S,\alpha) = 2$), then $p(\alpha) = |C \cap S|$; otherwise, $p(\alpha) = \max\{ p(\beta) \mid \beta \in N_{J}(\alpha) \cap I_{J}(x_S,\alpha) \}$. Note that, we can compute all those values $p(\alpha)$ during a BFS with no significant computational overhead.
Furthermore, we claim that $\forall v \in V, p(v) = |Pr_G(v,S)|$.
Indeed, by induction, $p(v) = \max\{ |C \cap S| \mid x_C \in X_{\cal C} \cap I_{J}(v,S)\}$. 
We recall our earlier characterization of the shortest $vS$-paths in $J$ as those obtained from the shortest $vS$-paths in $G$ by adding a maximal clique between every two consecutive vertices.
As a result, $p(v)$ is the number of vertices in $S$ that are 
at distance exactly $dist_G(v,S)$ from $v$, that is exactly $|Pr_G(v,S)|$.  

Finally, we recursively define $g(\alpha), \forall \alpha \in V(J) \setminus N^2_J[x_S]$ as follows (recall that $N_J^2[x_S] = S \cup \{ x_C \mid C \cap S \neq \emptyset \}$):
if $\alpha \in V$ and $\alpha \in N_G(S)$ (equivalently, $dist_{J}(x_S,\alpha) = 3$), then $g(\alpha) = \alpha$; otherwise, $g(\alpha) \in \{ g(\beta) \mid \beta \in N_{J}(\alpha) \cap I_{J}(x_S,\alpha) \land p(\alpha) = p(\beta) \}$.
Again, we can compute all those values $g(\alpha)$ during a BFS with no significant computational overhead.
Furthermore, it also follows from our characterization of shortest $vS$-paths in $J$ that we have, $\forall v \notin S$, $g(v)$ is a gate of $v$.  
\end{proof}

\begin{lemma}\label{lem:split-size}
For a clique-tree $T$ of a given chordal graph $G$, let $H$ be the split graph constructed as in Section~\ref{sec:chordal-reduction}.
The sparse representation of the split graph $H$ can be computed in ${\cal O}(w(T))$ time.
Furthermore, if $U \subseteq V(H)$ is the stable set of $H$, then $|U| = {\cal O}(|V(T)|)$ and $\ell(H) = {\cal O}(w(T))$.
\end{lemma}

\begin{proof}
After Lemma~\ref{lem:gate-with-tree}, we need to select a subset of pairwise non-adjacent gates in order to construct $U$.
For that, we consider all the maximal cliques $C$ sequentially.
If $C$ contains at least two gates, then we suppress all the gates in $C$ but one with minimum metric projection on $S$.
Overall, this post-processing also takes time ${\cal O}(w(T))$.
Doing so there is at most one gate selected per maximal clique, {\it i.e.}, $|U| \leq |V(T)|$.
Furthermore, we have $\ell(H) = \sum_{u \in U} deg_H(u) \leq \sum_{C} |C| = w(T)$.
We end up adding $a,b$ and the edges incident to these two vertices and the stable set $U$, that takes total time ${\cal O}(|U|) = {\cal O}(|V(T)|)$. 
\end{proof}

Finally, in order to complete the proof of Theorem~\ref{thm:chordal-reduction}, let $(H_i)$ be the family of split graphs considered for a given recursive step of the reduction.
\begin{itemize}
\item Let us assume that for every $i$, we can compute $diam(H_i)$ in ${\cal O}(|E(H_i)|^a\cdot|V(H_i)|)$ time.
By construction, $\max_i |E(H_i)| = {\cal O}(m+n)$.
Furthermore, by Lemma~\ref{lem:split-size}, the gate vertices in the stable sets of $H_i$ sum up to ${\cal O}(|V(T)|) = {\cal O}(n)$.
Since, in addition, the maximal cliques of the split graphs $H_i$ are pairwise different maximal cliques of $G$ (each augmented by two new vertices), we get $\sum_i |V(H_i)| = {\cal O}(w(T)) = {\cal O}(m)$.
As a result, computing $diam(H_i)$ for all $i$ takes total time ${\cal O}(m^{1+a})$.
\item In the same way let us assume that for every $i$, we can compute $diam(H_i)$ in ${\cal O}(\ell(H_i)\cdot|V(H_i)|^b)$ time.
By Lemma~\ref{lem:split-size}, $\sum_i \ell(H_i) = {\cal O}(w(T)) = {\cal O}(n+m)$.
Furthermore, $\max_i |V(H_i)| \leq n$.
As a result, computing $diam(H_i)$ for all $i$ takes total time ${\cal O}(mn^b)$.
\end{itemize}

\subsection{Application: Approximating all eccentricities}\label{sec:ecc-chordal}

It was proved in~\cite{Dra19} that for all chordal graphs, an additive $+2$-approximation of all eccentricities can be computed in total linear time.
Using our previous reduction from Section~\ref{sec:chordal-reduction}, we improve this result to an additive $+1$-approximation, but at the price of a logarithmic overhead in the running time.

\begin{theorem}\label{thm:ecc-chordal}
For every $n$-vertex $m$-edge chordal graph, we can compute an additive $+1$-approximation of all eccentricities in total ${\cal O}(m\log{n})$ time.
\end{theorem}

The remainder of Section~\ref{sec:ecc-chordal} is devoted to the proof of this above theorem.
For that, we need to carefully revisit the reduction from Section~\ref{sec:chordal-reduction}.
In what follows, let $G$ be chordal and let $T$ be a fixed clique-tree of $G$. We can assume that $G$ has at least two vertices. 
\begin{itemize}
    \item If $T$ is reduced to a single node, or equivalently, $G$ is a complete graph, every vertex has eccentricity equal to $1$. In the same way if $T$ is reduced to exactly two nodes, then $G$ is the union of two crossing complete subgraphs $C_1$ and $C_2$. Furthermore, every vertex of $C_1 \cap C_2$ has eccentricity equal to $1$, whereas every vertex of the symmetric difference $C_1 \Delta C_2$ has eccentricity equal to $2$.
    \item Otherwise, let $S \in V(T)$ be a weighted centroid of $T$. By Lemma~\ref{lem:weighted-centroid}, the clique $S$ can be computed in time ${\cal O}(|V(T)|) = {\cal O}(w(T))$ if we are given $T$ in advance. 
    \begin{itemize}
        \item Let $T_1,T_2,\ldots,T_{\ell}$ be the components of $T \setminus \{S\}$, and for every $1 \leq i \leq \ell$ let $V_i := \left(\bigcup V(T_i)\right) \setminus S$.
We recall~\cite{BlP93} that the sets $V_1,V_2,\ldots,V_{\ell}$ are exactly the connected components of $G \setminus S$, and that the closed neighbourhoods $N_G[V_i], 1 \leq i \leq \ell$ induce distance-preserving subgraphs of $G$. As in Section~\ref{sec:chordal-reduction}, we denote these subgraphs by $G_1,G_2,\ldots, G_{\ell}$.
We apply our reduction recursively on each of these subgraphs $G_i$.
Doing so, for every $i$ and every vertex $v_i \in V_i$ we get an additive $+1$-approximation of $e_{G_i}(v_i)$.
        \item Then, for every vertex $v \notin S$, we compute $dist_G(v,S)$ and a gate $g_S(v)$. By Lemma~\ref{lem:gate-with-tree} this can be done in total ${\cal O}(w(T))$ time if $T$ is given in advance. Furthermore, notice that since $S$ is a clique, for every $s \in S$ we have that $\max_{v \in V}dist_G(v,S)$ is an additive $+1$-approximation of $e_G(s)$. 
        \item Finally, we observe that for every $i$ and every vertex $v_i \in V_i$ we have:
        $$e_G(v_i) = \max \{e_{G_i}(v),\max_{u \notin V_i} dist_G(v_i,u)\}.$$
        As in Section~\ref{sec:chordal-reduction} we define $\forall i, d_i := \max_{v_i \in V_i} dist(v_i,S)$. If we compute the two largest values amongst the $d_i$'s, then we can compute $\forall i, e_i := \max_{j \neq i} d_j$. We are done as for every $v_i \in V_i$, $dist_G(v_i,S) + e_i$ is an additive $+1$-approximation of $\max_{u \notin V_i} dist_G(v_i,u)$.
    \end{itemize}
\end{itemize}

As already observed in Section~\ref{sec:chordal-analysis}, there are ${\cal O}(\log{w(T)}) = {\cal O}(\log{n})$ recursion levels. Since at any level, the clique-trees of all the subgraphs considered form a family of pairwise disjoint subtrees of $T$, each step can be done in ${\cal O}(w(T))$ time.
Therefore, the total running time is in ${\cal O}(w(T)\log{n}) = {\cal O}(m\log{n})$~\cite{BlP93}.

\subsection{Application to chordal graphs of constant VC-dimension}\label{sec:vc-dim}

The VC-dimension of a graph $G$ is the largest cardinality of a subset $S$ such that $\{ N[v] \cap S \mid v \in V \} = 2^S$ (we say that $S$ is shattered by $G$). For instance, interval graphs have VC-dimension at most two~\cite{BLLP+15}. We now apply the reduction of Theorem~\ref{thm:chordal-reduction} so as to prove the following result:

\begin{theorem}\label{thm:vc-dim}
For every $d > 0$, there exists a constant $\eta_d \in (0;1)$ such that in ${\cal O}(mn^{1-\eta_d})$ time, we can compute the diameter of any chordal graph of VC-dimension at most $d$.
\end{theorem}

\begin{proof}
If a split graph $H$ has VC-dimension at most $d'$, then we can compute its diameter in truly subquadratic time ${\cal O}(\ell(H)\cdot|V(H)|^{1-\varepsilon_{d'}})$, for some $\varepsilon_{d'} \in (0;1)$~\cite[Theorem 1]{DHV19+a}.
As a result it is sufficient to prove that all the split graphs $H_i$, which are output by the reduction of Theorem~\ref{thm:chordal-reduction}, have a VC-dimension upper bounded by a function of $d$.
We observe that every such $H_i$ is obtained from an induced subgraph $H_i'$ of $G$ by adding two new vertices $a$ and $b$ (see Section~\ref{sec:chordal-reduction}).
Since $G$ has VC-dimension at most $d$, so does $H_i'$.
Then, let $S$ be a largest subset shattered by $H_i$.
We can extract from $S$ a maximal shattered subset $S' \subseteq V(H_i')$ ({\it i.e.}, not containing $a$ and $b$).
In particular, $|S'| \geq |S| - 2$.
Furthermore, since $S'$ is shattered by $H_i$, $\{ N[v] \cap S' \mid v \in V(H_i) \} = 2^{S'}$ holds.
It implies that $|\{ N[v] \cap S' \mid v \in V(H_i') = V(H_i) \setminus \{a,b\} \}| \geq 2^{|S'|} - 2$.
By the Sauer-Shelah-Perles Lemma~\cite{Sau72,She72}, we also have $|\{ N[v] \cap S' \mid v \in V(H_i') \}| = {\cal O}(|S'|^d)$, which implies $|S'| = {\cal O}(d\log{d})$. 
Consequently, every $H_i$ has VC-dimension in ${\cal O}(d\log{d})$. 
\end{proof}

We left open whether there exist other subclasses of chordal graphs for which we can use our techniques in order to compute the diameter in truly subquadratic time.
\bibliographystyle{abbrv}
\bibliography{biblio-Helly}

\end{document}